\documentclass[a4paper,10pt]{article}
\usepackage{amsmath, amsthm, amssymb}
\usepackage[pdftex]{graphicx}
\usepackage{enumerate}
\usepackage{mathrsfs}
\usepackage{color}
\usepackage{cancel}
\usepackage{bm}

\usepackage[utf8]{inputenc}
\usepackage[nottoc,numbib]{tocbibind}

\def\XXint#1#2#3{{\setbox0=\hbox{$#1{#2#3}{\int}$}
    \vcenter{\hbox{$#2#3$}}\kern-.5\wd0}}

\usepackage[misc]{ifsym}

\usepackage{bibspacing}
    \newcommand\email[1]{\_email #1\q_nil}
    \def\_email#1@#2\q_nil{%
      \href{mailto:#1@#2}{{\emailfont #1\emailampersat #2}}
    }
    \newcommand\emailfont{\sffamily}
    \newcommand\emailampersat{{\color{cyan}\small@}}
		
		\usepackage{hyperref}
\hypersetup{
    colorlinks=true,
    linkcolor=blue,          
    citecolor=blue,        
    filecolor=magenta,      
    urlcolor=cyan           
}
 
\def\({\left(}
\def\){\right)}

\newcommand{\eps}{\epsilon}

\newcommand{\be}{\begin{equation}}
\newcommand{\ee}{\end{equation}}
\newcommand{\bea}{\begin{eqnarray}}
\newcommand{\eea}{\end{eqnarray}}
\newcommand{\beann}{\begin{eqnarray*}}
\newcommand{\eeann}{\end{eqnarray*}}
\newcommand{\nnn}{\nonumber}

 \newtheorem{theorem}{Theorem}[section]

\newtheorem{remark}[theorem]{Remark}

\newtheorem{lemma}[theorem]{Lemma}

\numberwithin{equation}{section}

\begin{document}
\setcounter{page}{1} 
\noindent 

\title{{\footnotesize Manuscript submitted September 30, 2018; revised April 9, 2019}\\
\bf{Semi-classical analysis with new Galilean transformations for a Gross--Pitaevskii system
with non-zero conditions at infinity}}

\author{Qi Gao
\thanks{Department of Mathematics, School of Science, Wuhan University of Technology, Wuhan, 430070, P.R. China,
\email{gaoq@whut.edu.cn}}
 \and Chiun-Chang Lee
\thanks{Institute for Computational and Modeling Science, National Tsing Hua University,
 Hsinchu 30014, Taiwan, \email{chlee@mail.nd.nthu.edu.tw}}
 \and Tai-Chia Lin
 \thanks{Department of Mathematics, National Taiwan University, Taipei 10617, Taiwan, \email{tclin@math.ntu.edu.tw}}
  }

\thispagestyle{empty}
\date{}
\maketitle

\begin{abstract}
Recently, a rich variety of the micro-phenomena of the superfluid passing an obstacle has been observed in the binary mixture of rotating Bose--Einstein condensates (BECs). Among such phenomena, the interaction of dark--bright solitons is one of the most
important issues. In this work we investigate the semi-classical limit for a coupled system of Gross--Pitaevskii (GP) equations with rotating fields and trap potentials in a two-dimensional exterior domain, where the superfluid is non-vanishing at infinity.  We establish a new Galilean type transformation and follow the argument of the modulated energy functional (a Lyapunov type functional) in \cite{ll08,lz05} to control the propagation of mass densities and linear momenta of the solution via a compressible Euler equation with Coriolis force in a semi-classical regime. Moreover, the effect of the rotating field on the superfluid in the region far away from the obstacle is precisely described.

\bigskip

\noindent
{\footnotesize {\bf Keywords.} Rotating Bose--Einstein condensates; Dark--bright solitons; Gross--Pitaevskii equations; Semi-classical limit;\\ \hspace*{46pt}  Galilean transformation}

\medskip


\end{abstract}


\section{Introduction}

\noindent

A wider variety of Gross--Pitaevskii (GP) equations~\cite{ecp02,gp58,hh02,llbk89} have been broadly used as fundamental models
for describing superfluidity of Bose--Einstein condensates~(BECs).
In particular, such models can also be used to describe the micro-phenomena of the superfluid passing an obstacle~\cite{ba01, egk06, fpr92, gk08, jp01,mzy07,nfk08, skc07,uiw06}.
Recently, theories in rotating two-component BECs have stimulated new interest in the
coupled system of GP equations with rotating fields and trap potentials~\cite{amw12,imw04,llm05,ps2003bk}.
Hence, based on those physical phenomena and the related mathematical descriptions in~\cite{ll08, lz05, lz06, liu06},
we are interested in the situation that the superfluid is non-vanishing far away from the obstacle.

We consider such models in $2+1$ dimensional space-time domains.
Let $\Omega:=\mathbb{R}^2\setminus\mathcal{K}$ be an exterior domain, where $\mathcal{K}$ represents the obstacle which
is a simply connected compact subset of $\mathbb{R}^2$.
We shall investigate the semi-classical asymptotics for a coupled system of
GP equations with a rotating field~$A$ and a trap potential~$V$ in $[0,\infty)\times\overline{\Omega}$:
\begin{equation}\label{trap}
\left\{
\begin{array}{ll}
i\epsilon\partial_t \psi^{\epsilon}_1 =-{1\over2}\Delta^{\epsilon, \eta}_A \psi^\epsilon_1 +V \psi^\epsilon_1 +|\psi^\epsilon_1 |^2 \psi^\epsilon_1 +\gamma |\psi^\epsilon_2 |^2 \psi^\epsilon_1,&\quad\,\,\, x\in\Omega,\ t>0,\\
i\epsilon \partial_t \psi^{\epsilon}_2 =-{1\over2}\Delta^{\epsilon, \eta}_A \psi^\epsilon_2 +V \psi^\epsilon_2 +|\psi^\epsilon_2 |^2 \psi^\epsilon_2 +\gamma |\psi^\epsilon_1 |^2 \psi^\epsilon_2,&\quad\,\,\, x\in\Omega,\ t>0,\\
 \psi^{\epsilon}_k \big|_{t=0} = \psi^{\epsilon}_{k,0}(x):=\sqrt{\rho^\eps_{k,0}(x)}\exp\left({\frac{i}{\eps}S^\eps_{k,0} (x)}\right),&\quad\,\,\, x\in\Omega,\\
 \rho^{\epsilon}_{k,0}(x)\to {a_k},\quad\mathrm{and}\quad{S}^\eps_{k,0} (x)\to{U}^{\infty}\cdot {x}, &\quad\text{as}\,\, |x|\to\infty,\,\, t>0.
 \end{array}
\right.
\end{equation}
 On the boundary
 $\partial\Omega$, following \cite{hg74}
we impose the Neumann boundary conditions as below: 
\begin{equation}\label{0219-7}
\left.\frac{\partial\psi^\epsilon_k}{\partial{\vec{n}}}\right|_{\partial\Omega}=0,\ \ \text{for}\ t>0,\ k=1, 2, 
\end{equation}
where $\vec{n}(x)$ is the unit outward normal vector to $\partial\Omega$ at
$x\in\partial\Omega$.

For the model \eqref{trap}, $0<\eps\ll1$
 is a dimensionless parameter scaled by the Planck's constant $\hbar$ ($\eps$ is usually called the semi-classical parameter, see \cite{lw12} for the detailed dimensional formulation),
$i=\sqrt{-1}$, and
$\psi^\eps_k\equiv\psi^\eps_k (t,x)\in\mathbb{C}$
are complex-valued functions defined in $2+1$ dimensional space-time for $k=1,2$.
The rotating field $A\equiv{A}(t,x)\in\mathbb{R}^2$ is a vector-valued function, and
the operators $\nabla^{\eps,\eta}_A$ and $\Delta^{\eps, \eta}_A$  are defined by
\begin{align}\label{0219-1}
\begin{cases}
\nabla^{\eps,\eta}_A:=\eps\nabla-\eps^\eta (iA),\\
\Delta^{\eps, \eta}_A:=\nabla^{\eps, \eta}_A \cdot\nabla^{\eps, \eta}_A
=\eps^2\Delta-2\eps^{1+\eta}(iA\cdot\nabla)-\eps^{2\eta}|A|^2,
\end{cases}
\end{align}
where $\eta\geq0$, $\nabla=(\partial_{x_1}, \partial_{x_2})$,
 ``$\cdot$" denotes the inner product of vectors, and $|A|^2 =A\cdot A$.
Besides, $U^\infty$ is a constant two-vector, and $a_k \ge0$, $k=1, 2$,
 $\gamma\ge1$ are constants. All of those parameters are  independent of $\eps$.

We stress that the initial data $\psi^{\epsilon}_{k,0}$'s presented in \eqref{trap} come from the concept of
the standard Madelung transformation~\cite{m27}, where $\rho^\eps_{k,0}=|\psi^{\epsilon}_{k,0}|^2$ is the mass density,
 and the dimensionless quantity $\frac{S^\eps_{k,0}}{\eps}\,(\sim\frac{S_{k,0}}{\hbar})=\mathbf{arg}\psi^{\epsilon}_{k,0}:=\arctan\frac{\Im\psi^{\epsilon}_{k,0}}{\Re\psi^{\epsilon}_{k,0}}$ is the phase function. Here $\Re$ and $\Im$ represent the real and imaginary parts, respectively. We will frequently use these two notations without further comments.

Among the phenomena of two-component BECs, the interaction of dark--bright solitons is one of the most important (cf. \cite{ba01,ychkeafc2011}); see also, a survey paper~\cite{kf2016} and references therein. Motivated by these investigations, we focus on the competition of the interaction forces between the dark- and bright-soliton components and illustrate the semi-classical limit of (\ref{trap}) with the boundary condition~(\ref{0219-7}). Accordingly, in \eqref{trap} we assume that each mass density~$\rho^\eps_{k,0}$ is preserved as a constant~$a_k$ and
\begin{align}\label{0302a}
(a_1,a_2)=(1,0)\,\,\mathrm{or}\,\,(0,1).
\end{align}
Note also that the superfluid is non-vanishing far away from the obstacle since $a_1+a_2\neq0$. Furthermore, when far away from the obstacle, we impose a simplified assumption on the initial data $S^\eps_{k,0}$ such that both phase functions~$S^\eps_{1,0}$ and $S^\eps_{2,0}$ act as a same constant velocity~$U^\infty$. In view of the related physical background, it is expected an interesting phenomenon for the case that $S^\eps_{k,0}(x)\stackrel{|x|\to\infty}{-\!\!\!-\!\!\!-\!\!\!-\!\!\!-\!\!\!\longrightarrow}{U}^\infty_k\cdot{x}$ with ${U}^\infty_1\neq{U}^\infty_2$, however the mathematical rigorous analysis seems to be of a great challenge. We will make a brief discussion on this issue in Section \ref{sec-remark}.

The regularity issue for solutions $(\psi^{\epsilon}_1,\psi^{\epsilon}_2)$ to the system $\eqref{trap}$
with boundary conditions~(\ref{0219-7})
depends on the {\it smoothness} of $A$, $V$ and initial data $\psi^{\epsilon}_{k,0}$'s.
To get sufficient regularity of $(\psi^{\epsilon}_1,\psi^{\epsilon}_2)$
so that we can study its asymptotic limits,
for each $\epsilon>0$ we need to make the following assumptions:

\begin{itemize}
\item The initial data $\psi^{\epsilon}_{k,0}(x)$'s satisfy
\begin{equation}\label{0225-1}
 \psi^{\epsilon}_{k,0}(x)-\sqrt{a_k}\exp\left(\frac{i}{\epsilon}U^{\infty}\cdot x\right)\in H^4(\Omega;\mathbb{C})\cap\mathcal{C}^\infty (\overline{\Omega};\mathbb{C}),
\end{equation}
and are compatible with
the Neumann boundary conditions \eqref{0219-7}.


\item $A\equiv\,A(t,x)\in\mathbb{R}^2$ satisfies
\begin{equation}\label{0219-2}
A-A^\infty \in\mathcal{C}^2 ([0,\infty);  H^3 (\Omega; \mathbb{R}^2)
\cap\mathcal{C}^\infty(\overline{\Omega}; \mathbb{R}^2)),
\end{equation}
\begin{equation}\label{0219-4}
A\cdot\vec{n}=0\ \ \text{on}\ \ \partial\Omega,\quad\mathrm{for}\,\, t\ge0,
\end{equation}
and the divergence free condition
\begin{equation}\label{0219-3}
\text{div}A:=\nabla\cdot A=0\ \ \text{in}\ \ \Omega,\quad\mathrm{for}\,\, t\ge0, 
\end{equation}
where $A^\infty$ is a constant two-vector.
\item $V\equiv V(t,x)\in\mathbb{R}$ is a real-valued function satisfying
\begin{equation}\label{0219-5}
V-V^\infty\in\mathcal{C}^1 ([0,\infty);H^2 (\Omega; \mathbb{R})\cap\mathcal{C}^\infty (\overline{\Omega};\mathbb{R})), 
\end{equation}
where $V^\infty$ is a scalar constant.
\end{itemize}

 Let us here stress the meaning of the divergence free condition~(\ref{0219-3}). Indeed, (\ref{0219-3}) verifies that $A=\nabla\times\vec{E}$ can be expressed as rotation of a vector field~$\vec{E}$. In this situation, $\vec{E}$ is referred to as a vector potential~(see, e.g., \cite[Section~7.1.3.3]{r2017}). In particular, when $\mathcal{K}$ is an unit disk centered at the origin, a typical example of rotating field is introduced as below:
\begin{align*}
A(t,x=(x_1,x_2))=\left\{
\begin{array}{ll}
\omega_0 (t)(-x_2,x_1)&\text{in}\ \ [0,\infty)\times({B}_{R_1}\setminus\mathcal{K}),\\
\omega_0 (t)A^{\infty} &\text{in}\ \ [0,\infty)\times(\Omega \setminus B_{R_2}),
\end{array}
\right.
\end{align*}
and $A(t,x)$ is sufficiently smooth in $[0,\infty)\times(\overline{B_{R_2}}\setminus{B_{R_1}})$ so that \eqref{0219-2} holds,
where $\omega_0(t)-1\in\mathcal{C}^2([0,\infty))$, and $B_{R_i}$ is a disk of radius $R_i$ centered at the origin in $\mathbb{R}^2$, $i=1,2$, and $1\ll {R_1} < R_2$. With this definition, it guarantees that $A(t,x)$ is well-defined in $\mathbb{R}^2$. Besides, it also satisfies the condition \eqref{0219-4} and \eqref{0219-3}. Due to \cite{ll08}, it is expected that the rotation term $A$ will bring in some interesting phenomenon and make the mathematical analysis meaningful and challenging.

Under the above assumptions, we can apply the argument of the high-order energy estimates proposed in the Appendix of \cite{lz05}
to the model~\eqref{trap} with the boundary conditions~\eqref{0219-7},
and get the global existence, uniqueness and regularity of $(\psi^{\epsilon}_1,\psi^{\epsilon}_2)$ as follows (see Remark~\ref{rk-1123}):
\begin{align}
\partial_t^j\partial_x^{\alpha}\left(\psi_1^\eps (t,x)-\sqrt{a_1}
e^{\frac{i}{\eps}\left\{U^\infty\cdot x-\left[\frac{1}{2}|U^\infty-\eps^{\eta}A^{\infty}|^2+V^{\infty}+(1+(\gamma-1)a_2)\right]t\right\}}\right)\in{L}^{\infty}([0,T];H^{4-2j-|\alpha|}(\Omega;\mathbb{C})),\label{schrod-1123-1}\\
\partial_t^j\partial_x^{\alpha}\left(\psi_2^\eps (t,x)-\sqrt{a_2}
e^{\frac{i}{\eps}\left\{U^\infty\cdot x-\left[\frac{1}{2}|U^\infty-\eps^{\eta}A^{\infty}|^2+V^{\infty}+(1+(\gamma-1)a_1)\right]t\right\}}\right)\in{L}^{\infty}([0,T];H^{4-2j-|\alpha|}(\Omega;\mathbb{C})),\label{schrod-1123-2}
\end{align}
for any $T<\infty$ and $0\leq2j+|\alpha|\leq{4}$.

Before introducing our main result,
we shall trace back to the work~\cite{lz05} for a rigorous mathematical study on a single GP model
which investigates a superfluid passing an obstacle. There they considered
the semi-classical limit (as $\eps\downarrow0$) for
\begin{equation}\label{lz}
i\eps \partial_t \psi^\eps =-\frac{\eps^2}{2}\Delta\psi^\eps +(|\psi^\eps |^2 -1)\psi^\eps,\quad\mathrm{for}\,\,x\in\Omega,\,\,t>0,
\end{equation}
with initial data~$\psi^\eps_0(x)=\sqrt{\rho^\eps_0 (x)}\exp\left({i\over\eps}S^\eps_0 (x)\right)$ and Neumann boundary condition~$\frac{\partial\psi^\eps}{\partial\vec{n}}=0$ on $\partial\Omega$.
Note also that $\Omega$ is an exterior domain in $\mathbb{R}^2$.
When far away from the obstacle in particular,
 the mass density~$\rho^\eps_0$ is non-vanishing, and the phase function~$S^\eps_0$ acts as a constant velocity~$u^\infty$
so that $\psi^\eps_0(x)-\exp\left(\frac{i}{\eps}u^\infty\cdot x\right)\in{H}^s(\Omega)$ for some $s\geq3$.
As a consequence, for each $\eps>0$, they obtained the existence and uniqueness of solution~$\psi^\eps$ satisfying
\begin{align}\label{inbehav}
\partial_t^j\partial_x^{\alpha}\left(\psi^\eps (t,x)-
e^{\frac{i}{\eps}\left(u^\infty\cdot x-\frac{|u^\infty|^2}{2}t\right)}\right)\in{L}^{\infty}([0,T];H^{s-2j-|\alpha|}(\Omega;\mathbb{C})),
\end{align}
for any $T<\infty$ and $0\leq2j+|\alpha|\leq{s}$. Furthermore, by making appropriate assumptions on $\rho^\eps_0$ and $S^\eps_0$ (see (A1)--(A3) in \cite{lz05} for the precise statements),
they established the semi-classical limits for the mass density~$\rho^\eps(t,x):=|\psi^\eps(t,x)|^2$
and the linear momenta ${J}^{\eps}(t,x):=\rho^\eps(t,x)\nabla{S}^\eps(t,x)=\eps\Im(\overline{\psi^\eps(t,x)}\nabla\psi^\eps(t,x))$ 
in the sense that for any $T\in(0,T^*)$,
\begin{align}\label{conv0302}
||\rho^\eps-\rho||_{L^2(\Omega)}+||J^\eps-\rho{u}||_{L^1_{\mathrm{loc}}(\Omega)}\to0\ \ \mathrm{uniformly\,\,in}\,\,[0,T],
\end{align}
where $T^*$ is a finite time so that $(\rho,u)\in\bigcap_{j=0}^3\mathcal{C}^j([0,T^*),H^{3-j}(\Omega))$ uniquely solves the following compressible Euler equation
in local time
\begin{equation}\label{comeuler}
\left\{
\begin{array}{l}
\partial_t \rho+\text{div}(\rho u)=0,\hspace{88pt}{x\in\Omega},\,\,t\in(0,T^*),\\
\partial_t u+(u\cdot\nabla)u+\nabla\rho =0,\hspace{59pt}{x\in\Omega},\,\,t\in(0,T^*),\\
u(0,x)=u_0 (x),\ \  \rho(0,x)=\rho_0 (x),\hspace{23pt}{x\in\Omega},
\end{array}
\right.
\end{equation}
with the slip boundary condition~$u\cdot\vec{n}|_{\partial\Omega}=0$.

We shall point out that (\ref{conv0302}) and (\ref{comeuler}) are obtained from two major ideas. This is useful for studying the asymptotics of solutions to the model (\ref{trap}). For the first idea, let us review the connection between (\ref{lz}) and (\ref{comeuler}) from the physical framework. In the situation that there is no superfluid at infinity, we may apply the
standard Madelung transformation to (\ref{lz}). By separating the real and imaginary parts and introducing $u^\eps=\nabla{S}^\eps$,
(\ref{lz}) is transformed into a quantum hydrodynamic system (also called the quantum Euler equation)
\begin{equation}
\left\{
\begin{array}{l}
\partial_t \rho^\eps+\text{div}(\rho^\eps {u}^\eps)=0,\\
\partial_t u^\eps+(u^\eps\cdot\nabla)u^\eps+\nabla\rho^\eps =\frac{\eps^2}{2}\nabla\left(\frac{\Delta\sqrt{\rho^\eps}}{\sqrt{\rho^\eps}}\right),
\end{array}
\right.\notag
\end{equation}
where $\frac{\eps^2}{2}\nabla\left(\frac{\Delta\sqrt{\rho^\eps}}{\sqrt{\rho^\eps}}\right)$ is the so-called quantum pressure (see \cite{m27}).
Due to the Feshbach resonance (see \cite{pgkb04}),
the quantum pressure can be neglected as $\eps$ approaches zero.
Consequently,
 $(\rho^\eps,u^\eps)$ may formally approximate to $(\rho,u)$ which is the solution of (\ref{comeuler}) with $u^\infty=\vec{\boldsymbol{0}}$. However, if the superfluid with an initial data $\nabla{S}^\eps_0(x)$ is non-vanishing at infinity,
the behavior of the superfluid at infinity seems nontrivial as time $t$ goes forward, and this is not easy to determine.
In this case, the Madelung transformation is not enough to describe the exact asymptotics of $\nabla{S}^\eps(t,x)$
at infinity. To get a clear picture on the behavior of $\nabla{S}^\eps$,
a crucial ingredient is the {\it Galilean invariance} (see \cite{ssbk99}). 
Lin and Zhang in \cite{lz05} introduced
  this concept and
used the {\it Galilean transformation}~$\widetilde{\mathcal{G}}_{u^\infty}$ for $\psi^\eps$:
\begin{align}\label{galilean-iv}
\widetilde{\mathcal{G}}_{u^\infty}[\psi^\eps(t,x)]:=\psi^\eps(t,x+u^\infty{t})e^{-\frac{i}{\eps}\left(u^\infty\cdot{x}+\frac{|u^\infty|^2}{2}t\right)}.
\end{align}
Denote by $\phi^\eps(t,x)\equiv\widetilde{\mathcal{G}}_{u^\infty}[\psi^\eps(t,x)]$, $\phi^\eps$ then fits the same form as GP equation~(\ref{lz}). In other words, the GP equation \eqref{lz} is {\it formally} invariant under the Galilean transformation \eqref{galilean-iv}.
This is so-called the {\it Galilean invariance}. Note that under the transformation $\widetilde{\mathcal{G}}_{u^\infty}$, the initial data $\phi_0^\eps(x)\equiv\widetilde{\mathcal{G}}_{u^\infty}[\psi_0^\eps(x)]=\psi_0^\eps(x)e^{-\frac{i}{\eps}u^\infty\cdot{x}}\to1$; i.e., without the velocity effect as $|x|\to\infty$.
Furthermore, a standard regularity theorem allows $\phi^\eps(t,x)\to1$ as $|x|\to\infty$,
which implies
$\psi^\eps(t,x)=\phi^\eps(t,x-u^\infty{t})e^{\frac{i}{\eps}\left(u^\infty\cdot{x}-\frac{|u^\infty|^2}{2}t\right)}\to{e}^{\frac{i}{\eps}\left(u^\infty\cdot{x}-\frac{|u^\infty|^2}{2}t\right)}$ as $|x|\to\infty$ for $t>0$.

The second idea is about controlling the propagation of mass densities~$\rho^\eps$ and linear momenta~$J^\eps$.
The rigorous justification of the convergence for $(\rho^\eps,J^\eps)$ towarding $(\rho,\rho{u})$ relies on the suitable convergence of the initial data. The main approach for this convergence consists of the conservation laws and the so-called \textit{modulated energy functional}
\begin{align}\label{0302he}
H^{\eps}(t)=\frac{1}{2}\int_{\Omega}\left|(\eps\nabla-iu)\psi^{\eps}\right|^2dx+\frac{1}{2}\int_{\Omega}|\rho^{\eps}-\rho|^2dx,
\end{align}
where the concept of establishing (\ref{0302he})
comes from the Wigner transform (see \cite{z022,z021}) and the energy functional of $\psi^\eps$.
On the other hand, an interesting viewpoint is that using the Madelung transformation~$\psi^\eps (x,t)=\sqrt{\rho^\eps (x,t)}e^{\frac{iS^\eps (x,t)}{\eps}}$, we can rewrite (\ref{0302he}) as $H^{\eps}(t)=\frac{1}{2}\int_{\Omega}\Big(\eps^2\left|\nabla\sqrt{\rho^\eps}\right|^2+\frac{1}{\rho^\eps}|J^\eps-\rho^\eps{u}|^2+|\rho^{\eps}-\rho|^2\Big)dx$.
 The authors in \cite{lz05} showed that
$H^{\eps}(t)$ has some similar behaviors to that of a Lyapunov functional satisfying a Gronwall-type growth estimate,
which gives us an intuition to study the convergence of the differences $|\rho^\eps-\rho|$
and $|J^\eps-\rho{u}|$ 
 under the appropriate assumptions on the initial data $H^\eps(0)$.

 The modulated energy functional \eqref{0302he} is particularly a powerful tool on studying the semi-classical limits of GP models.
Later on, Lin and Zhang studied the semi-classical limit of a coupled system of GP equations in a bounded domain in $\mathbb{R}^2$ without the rotating fields and trap potentials (see \cite{lz06}), while Lee and Lin studied the semi-classical limit of the same system as in \cite{lz06} with the rotating fields and trap potentials in a bounded domain (see \cite{ll08}). In \cite{lz06} and \cite{ll08}, analogous to the techniques in \cite{lz05}, the authors
established the corresponding modulated energy functionals and controlled the propagation of the total mass densities
and total linear momenta for binary superfluid.
Some related works using the concept of modulated energy functional can be also found in \cite{cr09, lll16, lw10, lw12} and references therein.

As for \cite{ll08}, we shall consider the model \eqref{trap}--(\ref{0219-7}) with different scales of the rotating field $\eps^\eta A$. This leads to a variety of phenomena that how rotation affects the compressibility of rotating two-component condensates. When $\eta$ is taken into consideration, i.e. $\eta>0$, we can remove the rotating field $A(t, x)$ from the limit compressible Euler equation; i.e., the limit system is given by
\begin{equation*}
\left\{
\begin{array}{ll}
\partial_t \rho+\text{div}(\rho u)=0,\\
\partial_t u+(u\cdot\nabla)u+\nabla\rho=-\nabla V,&\text{for}\ x\in\Omega,\ t\in[0,T_* ),\\
(\rho, u)(0, x)=(\rho_0 (x), u_0 (x)), &\text{for}\ x\in\Omega.
\end{array}
\right.
\end{equation*}
Therefore, to observe the effect of rotating field on the system, $\eta$ must be zero. On the other hand, although the rotating field $A(t,x)$ brings us the richness of phenomena, it also creates much more challenge for our problem. In order to deal with the behavior of superfluid at infinity, inspired by \cite{lz05}, we start by creating a new \textit{Galilean transformation} $\mathcal{G}_{u^{\eps,\eta,\infty}}$ introduced in Section~\ref{0920-sec-18} (cf. \eqref{0302ga})
 for $\eqref{trap}$:
\begin{align*}
&\mathcal{G}_{u^{\eps,\eta,\infty}}\left[\left(
\begin{array}{ccc}
\psi^{\eps}_1 (t,x) \\
\psi^{\eps}_2 (t,x)
 \end{array}
\right)\right]\notag\\[-0.7em]
\\[-0.7em]
&\qquad:=
\left(
\begin{array}{ccc}
\psi^\eps_1 (t,x+u^{\eps,\eta,\infty}t)e^{-\frac{i}{\eps}\left\{U^\infty\cdot x+\left[{1\over2}(|U^\infty|^2 -\eps^{2\eta}|A^\infty|^2 )-V^\infty-(1+(\gamma-1){{a}_2})\right]t\right\}} \\
\psi^\eps_2 (t,x+u^{\eps,\eta,\infty}t)e^{-\frac{i}{\eps}\left\{U^\infty\cdot x+\left[{1\over2}(|U^\infty|^2 -\eps^{2\eta}|A^\infty|^2 )-V^\infty-(1+(\gamma-1){{a}_1})\right]t\right\}}
 \end{array}
\right),\notag
\end{align*}
where
\begin{align}\label{cc0922}
u^{\eps,\eta,\infty}=U^\infty-\eps^{\eta}A^\infty.
\end{align}
By \eqref{0302a} and \eqref{0302ga},
 we then reduce the equations of $\psi^\eps_1$ and $\psi^\eps_2$ in \eqref{trap} into the system in terms of $\phi^\eps_1$ and $\phi^\eps_2$ (which are defined in \eqref{0302ga}). In particular, the resulting equations of $\phi^\eps_1$ and $\phi^\eps_2$ are crucial in our study because
their initial data~$\phi^{\eps}_k (0,x)=\psi^{\eps}_{k,0}(x)\exp\left(-\frac{i}{\eps}U^\infty\cdot x\right)\to\sqrt{a_k}$
as $|x|\to\infty$,
which means that the superfluid has no velocity effect at infinity.
Moreover, because of \eqref{0219-2}, for each $t>0$
we have $\widetilde{A}(t,x)-A^\infty\in{H}^3 (\Omega; \mathbb{R}^2)
\cap\mathcal{C}^\infty(\overline{\Omega}; \mathbb{R}^2)$
which is useful for dealing with the operator~$\Delta^{\epsilon, \eta}_{\widetilde{A}-A^\infty}$
in the whole domain~$\Omega$.
As a consequence, we can establish the corresponding conservation laws for energy, mass and linear momentum in Lemma~\ref{laws}.
We stress that the modified Galilean transformation and Madelung transformation play a crucial role in deriving the following compressible Euler equation corresponding to the limit of hydrodynamic equations of \eqref{trap} with $\eta=0$:
\begin{equation}\label{curlfield}
\left\{
\begin{array}{ll}
\partial_t \rho+\text{div}(\rho u)=0,\\
\partial_t u+(u\cdot\nabla)u+\text{curl}A\cdot u^{\bot}+\nabla\rho=-\partial_t A-\nabla V,&\text{for}\ x\in\Omega,\ t\in[0,T_* ),\\
(\rho, u)(0, x)=(\rho_0 (x), u_0 (x)), &\text{for}\ x\in\Omega,
\end{array}
\right.
\end{equation}
with the slip boundary condition
\begin{equation}\label{sbc}
\frac{\partial u}{\partial\vec{n}}\Big|_{\partial\Omega}=0,\ \ \text{and}\ \ \rho(t,x)\to1,\ \ u(t,x)\to u^\infty\ \ \text{as}\ \ |x|\to\infty,
\end{equation}
where
\begin{equation}\label{uinfitya}
u^\infty =U^\infty-A^\infty,
\end{equation}
i.e., (\ref{cc0922}) with $\eta=0$.

It is well known that the Wigner measure is a standard tool to study the semi-classical limit problem for single GP equation, especially for dealing the problem in the whole space (see \cite{z08}). However, for general coupled system of GP equations, the Wigner measure would be much more sophisticated and difficult to apply. Fortunately, Brenier introduced the method of modulated energy (functional) in 2000 (see \cite{b00}) which has since become a powerful tool in studying the singular limit problems of nonlinear PDEs.
Motivated by \cite{b00,ll08,lz05}, we may construct a modulated energy functional
\begin{align}\label{Hfunct}
H(t):=&\displaystyle {1\over2}\int_\Omega \sum^{2}_{k=1}\left|(\nabla^{\eps,\eta}_A -iu)\psi^\eps_k \right|^2 dx+\int_\Omega \left[{1\over2}(\rho^\eps_1 +\rho^\eps_2 -\rho)^2 +(\gamma-1)\rho^\eps_1 \rho^\eps_2 \right]dx\\
=&{{\eps^2}\over2}\int_{\Omega}\sum^2_{k=1}\left|\nabla\sqrt{\rho^\eps_k}\right|^2 dx+{1\over2}\int_\Omega \sum^{2}_{k=1}\frac{1}{\rho^\eps_k}\left|J^\eps_k -\rho^\eps_k u\right|^2 dx+\int_\Omega \left[{1\over2}(\rho^\eps_1 +\rho^\eps_2 -\rho)^2 +(\gamma-1)\rho^\eps_1 \rho^\eps_2 \right]dx,\notag
\end{align}
where $\rho^\eps_k$ (respectively, $J^\eps_k$) is the mass density (respectively, linear momenta) of $\psi^\eps_k$:
\begin{align}\label{0920-2018-5}
\rho^\eps_k:=|\psi^\eps_k|^2,\,\,\mathrm{and}\,\,J^\eps_k =(J^\eps_{k,1}, J^\eps_{k,2})\,\,\mathrm{with}\,\,
J^\eps_{k,j}:=\Im(\overline{\psi^\eps_k}\partial^{\eps,\eta}_{A_j}\psi^\eps_k ).
\end{align}
It can be viewed as a defect measure in studying weakly convergent sequences of solutions.
As mentioned previously, to see the effect of the rotating field, we shall set $\eta=0$.
In our technical analysis, we need the following assumptions on initial data with respect to $\eps$:
\begin{itemize}
\item[\textbf{(A1)}] $\displaystyle\eps^2\int_{\Omega}\left|\nabla\sqrt{\rho^\eps_{k,0}(x)}\right|^2dx\to0$ as $\eps\to0$, for $k=1, 2$;

\item[\textbf{(A2)}] Both $\Big(\big(\nabla S^\eps_{k,0} (x)-A_0 (x)\big)-u_0 (x)\Big)\sqrt{\rho^\eps_{k,0}(x)}$ and $\rho^\eps_{1,0}(x)+\rho^\eps_{2,0}(x)-\rho_0 (x)$ converge to zero in $L^2 (\Omega)$ as $\eps\to0$;

\item[\textbf{(A3)}] (\ref{0302a}) holds and $\displaystyle\int_{\Omega}\rho^\eps_{1,0}(x)\rho^\eps_{2,0}(x)dx\to0$ as $\eps\to0$.
\end{itemize}
Assumptions (A1)--(A3) assert that $H(0)\to0$ as $\eps\to0+$. Moreover, by
using the modified Galilean transformation \eqref{0302ga}, we can prove that $H(t)$ satisfies a Gronwall-type growth estimate, which implies $H(t)\to0$ as $\eps\to0+$ for $t$ in an interval of considerations. 
The main result is stated as follows.

\begin{theorem}\label{main}
Assume $\gamma\geq1$ and (\ref{0225-1})--(\ref{0219-5}).
For $\eta=0$, let $(\psi^\epsilon_1 (t,x), \psi^\epsilon_2 (t,x) )$ be the unique solution of \eqref{trap} with the Neumann boundary conditions \eqref{0219-7}, and $(\rho(t,x),u(t,x))$
be the solution of \eqref{curlfield}--\eqref{uinfitya}, where $0<\rho_0\in{H}^3(\Omega;\mathbb{R})$
and $u_0\in{H}^3(\Omega;\mathbb{R}^2)$. Then, under the assumptions (A1)--(A3),
 there exists a constant $T_* >0$ such that for any $T<T_*$
\begin{equation}\label{converge}
\left\{
\begin{array}{l}
(J^\epsilon_1 +J^\epsilon_2 -\rho u)(t,x)\longrightarrow0\ \ \text{in}\ \ L^\infty\left([0,T],L^1_{\text{loc}}(\Omega)\right),\\
(\rho^\epsilon_1 +\rho^\epsilon_2 -\rho)(t,x)\longrightarrow0\ \ \text{in}\ \ L^\infty\left([0,T], L^2 (\Omega)\right),
\end{array}
\right.
\end{equation}
as $\eps\to0$, where $\rho_k^{\eps}$ and
 $J^\epsilon_k =\Im(\overline{\psi^\epsilon_k}\nabla^{\epsilon,\eta}_A \psi^\epsilon_k )$ are defined in (\ref{0920-2018-5}).
Moreover, if $\gamma>1$, we have
\begin{align}\label{rho1124}
\sup_{t\in[0,T]}\int_{\Omega}\rho^\epsilon_1(t,x)\rho^\epsilon_2(t,x)\to0,\quad{as}\,\,\eps\to0.
\end{align}
\end{theorem}

\begin{remark}
When \eqref{rho1124} is satisfied, it signifies that there is no or very weak energy transfer between $\psi^\epsilon_1$ and $\psi^\epsilon_2$ . While, there is strong interaction between the two states when $\displaystyle\int_{\Omega}\rho^\epsilon_1(t,x)\rho^\epsilon_2(t,x)$ does not uniformly tend to zero in any finite time interval as $\epsilon$ goes to zero.
\end{remark}

\begin{remark}
Note that $A$ and $V$ are smooth functions.
By the standard methods (see \cite{b81,i87,lz05}), one may prove the local well-posedness
of \eqref{curlfield}-\eqref{uinfitya}, provided that $0<\rho_0\in{H}^3(\Omega;\mathbb{R})$
and $u_0\in{H}^3(\Omega;\mathbb{R}^2)$.
\end{remark}

The paper is organized as follows. The general modified Galilean transformation and conservation laws are covered in Section 2. In Section 3, after the construction of the modulated energy functional,
we give the proof of our main theorem. Finally, we give a concluding remark and pose a future project in Section 4.

\section{Modified Galilean transformation and conservation laws}
In this section we first introduce a variant of the classical Galilean transformation, then we construct the conservation laws of energy, mass and linear momentum of \eqref{trap}. Without loss of generality, we present our analysis with $a_1 =1$, $a_2 =0$ in the rest of this paper without further mention.

\subsection{Modified Galilean transformation}\label{0920-sec-18}

%
Recall that \eqref{lz} is formally invariant under the classical Galilean transformation \eqref{galilean-iv}. However, due to the existence of the rotating fields $A$ and the trap potential $V$, \eqref{galilean-iv} is NOT applicable for our system \eqref{trap}. This leads us to find a new form of Galilean transformation which we will present in \eqref{0302ga}. With this general {\it modified Galilean transformation}, we are able to deal with the rotating fields at infinity, and obtain the conservation laws of energy, mass and linear momentum.

We now define $u^{\eps,\eta,\infty}=U^\infty-\eps^\eta A^\infty$ and the general \textit{modified Galilean transformation} $\mathcal{G}_{u^{\eps,\eta,\infty}}$
 for the model $\eqref{trap}$ as follows:
 \begin{align}\label{0302ga}
&\mathcal{G}_{u^{\eps,\eta,\infty}}\left[\left(
\begin{array}{ccc}
\psi^{\eps}_1 (t,x) \\
\psi^{\eps}_2 (t,x)
 \end{array}
\right)\right]\notag\\[-0.7em]
\\[-0.7em]
&\qquad:=
\left(
\begin{array}{ccc}
\psi^\eps_1 (t,x+u^{\eps,\eta,\infty}t)e^{-\frac{i}{\eps}\left\{U^\infty\cdot x+\left[{1\over2}(|U^\infty|^2 -\eps^{2\eta}|A^\infty|^2 ){-V^\infty}-(1+(\gamma-1){a}_2)\right]t\right\}} \\
\psi^\eps_2 (t,x+u^{\eps,\eta,\infty}t)e^{-\frac{i}{\eps}\left\{U^\infty\cdot x+\left[{1\over2}(|U^\infty|^2 -\eps^{2\eta}|A^\infty|^2 ){-V^\infty}-(1+(\gamma-1){a}_1)\right]t\right\}} \\
 \end{array}
\right)\notag\\
&\qquad:=\left(
\begin{array}{c}
\phi^{\eps}_1 (t,x)\\
\phi^{\eps}_2 (t,x)
\end{array}
\right),\notag
\end{align}
where $A^\infty$ and $U^\infty$ are defined before. Moreover, with the choices of $a_1$ and $a_2$ in \eqref{0302ga}, it will give us a clear view of the connection between $\phi^\epsilon_k$ and $a_k$, $k=1, 2$. For convenience, we introduce the following notations:
\begin{equation*}
\left\{
\begin{array}{l}
M_1^{\eps} :={1\over2}\left(|U^\infty|^2 -\eps^{2\eta}|A^\infty|^2 \right)-V^\infty-\left[1+(\gamma-1){a}_2 \right]\\
M_2^{\eps} :={1\over2}\left(|U^\infty|^2 -\eps^{2\eta}|A^\infty|^2 \right)-V^\infty-\left[1+(\gamma-1){a}_1 \right].
\end{array}
\right.
\end{equation*}
Therefore, we can rewrite $\mathcal{G}_{u^{\eps,\eta,\infty}}$ in \eqref{0302ga} as follows:
\begin{equation}
\mathcal{G}_{u^{\eps,\eta,\infty}}\left[\left(
\begin{array}{ccc}
\psi^{\eps}_1 (t,x) \\
\psi^{\eps}_2 (t,x)
 \end{array}
 \right)\right]:=
\left(
\begin{array}{c}
\psi^\eps_1 (t,x+u^{\eps,\eta,\infty}t)e^{-\frac{i}{\eps}\left\{U^\infty\cdot x+M_1^{\eps} t\right\}} \\
\psi^\eps_2 (t,x+u^{\eps,\eta,\infty}t)e^{-\frac{i}{\eps}\left\{U^\infty\cdot x+M_2^{\eps} t\right\}}
 \end{array}
\right)
:=\left(
\begin{array}{c}
\phi^{\eps}_1 (t,x)\\
\phi^{\eps}_2 (t,x)
\end{array}
\right).
\end{equation}

By a direct calculation, we have that
\begin{equation}\label{phi1t}
\partial_t \phi^\eps_1 =\left[\partial_t\psi^\eps_1 +u^{\eps,\eta,\infty}\nabla\psi^\eps_1 -i\eps^{-1}M_1^{\eps} \psi^\eps_1\right]e^{-i\eps^{-1}\left\{U^\infty\cdot x+M_1^{\eps} t\right\}}\ ,
\end{equation}
\begin{equation}\label{gradientphi}
\nabla\phi^\eps_1=(\nabla\psi^\eps_1 -i\eps^{-1}U^\infty \psi^\eps_1)e^{-i\eps^{-1}\left\{U^\infty\cdot x+M_1^{\eps} t\right\}}\ ,
\end{equation}
\begin{equation}\label{laplacephi}
\Delta\phi^\eps_1 =(\Delta\psi^\eps_1-2i\eps^{-1}U^\infty \nabla\psi^\eps_1-\eps^{-2}|U^\infty|^2 \psi^\eps_1 )e^{-i\eps^{-1}\left\{U^\infty\cdot x+M_1^{\eps} t\right\}}\ .
\end{equation}
Since $\psi^\eps_1$ satisfies the $\psi^\eps_1$-equation in \eqref{trap}, it follows from \eqref{phi1t}-\eqref{laplacephi} that
\begin{multline}\label{phi1gteqn}
i\eps\partial_t \phi^\eps_1 =-{1\over2}\Delta^{\eps,\eta}_{\widetilde{A}-A^\infty}\phi^\eps_1 -\eps^\eta u^{\eps,\eta,\infty}(\widetilde{A}-A^\infty)\phi^\eps_1 \\
+(\widetilde{V}-V^\infty)\phi^\eps_1 +(|\phi^\eps_1 |^2 +|\phi^\eps_2 |^2 -1)\phi^\eps_1 +(\gamma-1)(|\phi^\eps_2|^2 -{a}_2 )\phi^\eps_1,
\end{multline}
where $\widetilde{A}=A(t,x+u^{\eps,\eta,\infty}t)$, $\widetilde{V}=V(t,x+u^{\eps,\eta,\infty}t)$.

Similarly, we can transform the original equation of $\psi^\eps_2$ in \eqref{trap} into the following:
\begin{multline}\label{phi2gteqn}
i\eps\partial_t \phi^\eps_2 =-{1\over2}\Delta^{\eps,\eta}_{\widetilde{A}-A^\infty}\phi^\eps_2 -\eps^\eta u^{\eps,\eta,\infty}(\widetilde{A}-A^\infty)\phi^\eps_2 \\
+(\widetilde{V}-V^\infty)\phi^\eps_2 +(|\phi^\eps_1 |^2 +|\phi^\eps_2 |^2 -1)\phi^\eps_2 +(\gamma-1)(|\phi^\eps_1|^2 -{a}_1 )\phi^\eps_2.
\end{multline}

Hence we obtain the {\it new} system which consists of the two equations \eqref{phi1gteqn} and \eqref{phi2gteqn} in terms of $\phi^\eps_1$ and $\phi^\eps_2$ under the modified Galilean transformation:
\begin{equation}\label{phisystm}
\left\{
\begin{array}{l}
 i\eps\partial_t \phi^\eps_1=-{1\over2}\Delta^{\eps,\eta}_{\widetilde{A}-A^\infty}\phi^\eps_1 -\eps^\eta u^{\eps,\eta,\infty}(\widetilde{A}-A^\infty)\phi^\eps_1\\
\qquad\qquad\qquad+(\widetilde{V}-V^\infty)\phi^\eps_1 +(|\phi^\eps_1 |^2 +|\phi^\eps_2 |^2 -1)\phi^\eps_1 +(\gamma-1)(|\phi^\eps_2|^2 -{a}_2 )\phi^\eps_1,\\
 i\eps\partial_t \phi^\eps_2=-{1\over2}\Delta^{\eps,\eta}_{\widetilde{A}-A^\infty}\phi^\eps_2 -\eps^\eta u^{\eps,\eta,\infty}(\widetilde{A}-A^\infty)\phi^\eps_2 \\
\qquad\qquad\qquad+(\widetilde{V}-V^\infty)\phi^\eps_2 +(|\phi^\eps_1 |^2 +|\phi^\eps_2 |^2 -1)\phi^\eps_2 +(\gamma-1)(|\phi^\eps_1|^2 -{a}_1 )\phi^\eps_2.\\
\phi^{\eps}_k (0,x)\to\sqrt{a_k}\quad\mathrm{as}\,\, \ |x|\to\infty.
\end{array}
\right.
\end{equation}
\begin{remark}\label{rk-1123}
Since $\phi^{\eps}_k (0,x)\to\sqrt{a_k}$ as $|x|\to\infty$,
under assumptions~\eqref{0225-1}, \eqref{0219-2}--\eqref{0219-3} and \eqref{0219-5},
the standard high-order energy estimates imply
\begin{align}\label{k12-1123}
\partial_t^j\partial_x^{\alpha}\left(\phi_k^\eps (t,x)-\sqrt{a_k}\right)\in{L}^{\infty}([0,T];H^{4-2j-|\alpha|}(\Omega;\mathbb{C})),\quad{k=1,2},
\end{align}
for any $T<\infty$ and $0\leq2j+|\alpha|\leq{4}$. Along with the Galilean transformation~\eqref{0302ga},
we get (\ref{schrod-1123-1}) and (\ref{schrod-1123-2}).
\end{remark}

\subsection{Conservation laws}

 When domain is bounded, the standard energy functional of \eqref{trap} is defined as follows (see \cite{ll08}):
\begin{align}\label{eg-17-1124}
E(\psi^\eps_1, \psi^\eps_2 )
                                 &={1\over2}\int_\Omega\sum^2_{k=1}\left|\nabla^{\eps,\eta}_A \psi^\eps_k \right|^2 dx +\int_\Omega (V+1)(|\rho^\eps_1 | +|\rho^\eps_2 | ) dx\notag\\
                                 &\qquad+{1\over2}\int_\Omega (\rho^\eps_1  +\rho^\eps_2 -1)^2 +(\gamma-1)\int_\Omega \rho^\eps_1 \rho^\eps_2 dx,
\end{align}
where $\rho^\eps_k =|\psi^\eps_k |^2$, $k=1, 2$.
However, when domain is unbounded and either $a_k\neq0$ or $U^\infty\neq\vec{\boldsymbol{0}}$ holds,
(\ref{eg-17-1124}) is not well-defined, and has to be modified.

Let us define a cut-off function $\chi (x)\in\mathcal{C}^{\infty}(\mathbb{R}^2 )$ with
\begin{equation}\label{chifctn}
\chi(x)=
\left\{
\begin{array}{ll}
0&\text{for}\ \ |x|\le R,\\
1&\text{for}\ \ |x|\ge 2R,
\end{array}
\right.
\end{equation}
such that $\|\nabla\chi\|_{L^\infty (\Omega)}\le{2\over R}$, where $R$ is sufficiently large such that $\Omega^{c}\subset B_{R}(0)$.
 By using $\chi(x)$ defined above, we can define a new energy functional as below:
\begin{align}\label{energychi-1}
\mathcal{E}(\psi^\eps_1, \psi^\eps_2 )
&=\int_\Omega\sum^2_{k=1}\left[{1\over2}\left|\nabla^{\eps,\eta}_A \psi^\eps_k \right|^2 +(V+1)\rho^\eps_k\right] (1-\chi)dx+\int_\Omega\sum^2_{k=1}{1\over2}\left|(\nabla^{\eps,\eta}_A -iu^{\eps,\eta,\infty}) \psi^\eps_k \right|^2 \chi dx \notag\\
&\quad+\int_\Omega\sum^2_{k=1} (V-V^\infty )\rho^\eps_k \chi dx+\int_\Omega {1\over2}(\rho^\eps_1 +\rho^\eps_2 -1)^2dx+(\gamma-1) \int_\Omega \rho^\eps_1 \rho^\eps_2 dx.
\end{align}
Here $u^{\eps,\eta,\infty}$ is defined in (\ref{cc0922}). 

Now we present the conservation laws.
\begin{lemma}\label{laws}
Let $e(t):=\mathcal{E}(\psi^\eps_1, \psi^\eps_2 )$, $J^\eps_{k,j}:=\Im(\overline{\psi^\eps_k}\partial^{\eps,\eta}_{A_j}\psi^\eps_k )$ and $J^\eps_k =(J^\eps_{k,1}, J^\eps_{k,2})$, then the following holds:
\vskip 0.2cm
\begin{itemize}
\item[(i)] Conservation of energy:
{\small\allowdisplaybreaks
\begin{align}\label{energy}
\frac{d}{dt}e(t)&=-\int_\Omega \sum^2_{k,j,l=1}\Re\left[(\partial^{\eps,\eta}_{A_j}\psi^\eps_k \overline{\partial^{\eps,\eta}_{A_l}\psi^\eps_k})u^{\eps,\eta,\infty}_l \partial_j \chi\right]+{{\eps^2}\over4}\int_\Omega \sum^2_{k=1}\Delta(u^{\eps,\eta,\infty}\cdot\nabla\chi)\rho^\eps_k\nnn\\
&\quad+\int_\Omega \sum^2_{k=1}\left[{1\over2}|u^{\eps,\eta,\infty}|^2 -1-V^\infty\right]J^\eps_k \cdot\nabla\chi-\eps^\eta \partial_t A\cdot J^\eps_k (1-\chi)\nnn\\
&\quad-\int_\Omega \sum^2_{k=1}\eps^\eta\left[\partial_t A +\mathrm{curl}A\cdot(u^{\eps,\eta,\infty})^\bot \right]\cdot(J^\eps_k -u^{\eps,\eta,\infty}\rho^\eps_k )\chi\nnn\\
&\quad+\int_\Omega \sum^2_{k=1}(\partial_t V)\rho^\eps_k +\int_\Omega \sum^2_{k=1}(\nabla V\cdot u^{\eps,\eta,\infty})\rho^\eps_k \chi\nnn\\
&\quad-\int_\Omega \left\{{1\over2}\left[(\rho^\eps_1 +\rho^\eps_2)^2 -1\right] +(\gamma-1)\rho^\eps_1 \rho^\eps_2\right\}\cdot(u^{\eps,\eta,\infty} \cdot\nabla\chi)\ ;
\end{align}}

\item[(ii)] Conservation of mass:
\begin{equation}\label{mass}
\partial_t \rho^\eps_k +\mathrm{div}J^\eps_k =0, \qquad \text{for}\ \ x\in\Omega,\ t>0,\ k=1, 2;
\end{equation}

\item[(iii)] Conservation of linear momentum:
\begin{multline}\label{momentum}
\sum^2_{k=1}\partial_t J^\eps_{k, j}+\sum^2_{k,l=1}\partial_l \left[\Re(\overline{\partial^{\eps,\eta}_{A_l}\psi^\eps_k}\cdot\partial^{\eps,\eta}_{A_j}\psi^\eps_k)\right]-{{\eps^2}\over4}\sum^2_{k=1}\partial_j (\Delta\rho^\eps_k )
+\sum^2_{k=1}(\partial_j V)\rho^\eps_k\\ +{1\over2}\sum^2_{k=1}\partial_j |\rho^\eps_k |^2 +\gamma\partial_j (\rho^\eps_1 \rho^\eps_2 )+\sum^2_{k=1} \eps^\eta(\partial_t A_j )\rho^\eps_k
+\sum^2_{k=1}(-1)^{j} \eps^\eta\mathrm{curl} A\cdot J^\eps_{k,j^*} =0,
\end{multline}
where $j^* =2$ if $j=1$, and $j^* =1$ if $j=2$.
\end{itemize}
\end{lemma}

\begin{proof}
(i) For convenience, we rewrite the equations of $\psi^\epsilon_1$, $\psi^\epsilon_2$ in \eqref{trap} into the following system:
\begin{equation}\label{transftrap}
\left\{
\begin{array}{l}
i\epsilon\partial_t \psi^\epsilon_1 =-{1\over2}\Delta^{\epsilon,\eta}_{A}\psi^\epsilon_1 +V\psi^\eps_1 +\psi^\epsilon_1 +(|\psi^\epsilon_1 |^2 +|\psi^\epsilon_2 |^2 -1)\psi^\epsilon_1 +(\gamma-1)|\psi^\epsilon_2 |^2 \psi^\epsilon_1,\\
i\epsilon\partial_t \psi^\epsilon_2 =-{1\over2}\Delta^{\epsilon,\eta}_{A}\psi^\epsilon_2 +V\psi^\eps_2 +\psi^\epsilon_2 +(|\psi^\epsilon_1 |^2 +|\psi^\epsilon_2 |^2 -1)\psi^\epsilon_2 +(\gamma-1)|\psi^\epsilon_1 |^2 \psi^\epsilon_2.
\end{array}
\right.
\end{equation}

We proceed our proof of part (i) into several steps.

\vskip 0.3cm

{\bf Step 1.} Multiply $\left[1-\chi(x)\right]\partial_t \psi^\eps_1$ to the complex conjugate of the $\psi^\epsilon_1$-equation in \eqref{transftrap} and
$\left[1-\chi(x)\right]\partial_t \overline{\psi^\eps_1}$ to the $\psi^\epsilon_1$-equation in \eqref{transftrap}. Then, summing up the two resulting equations and choosing the real part, we obtain that
{\allowdisplaybreaks
\begin{align}\label{Aoperator}
&-{1\over2}\int_\Omega (\Delta^{\eps,\eta}_A \psi^\eps_1 \cdot\partial_t \overline{\psi^\eps_1}+\overline{\Delta^{\eps,\eta}_A \psi^\eps_1}\cdot\partial_t \psi^\eps_1)(1-\chi)\nnn\\
&\qquad={1\over2}\int_\Omega \eps^2 (\partial_t |\nabla\psi^\eps_1 |^2)(1-\chi)-{1\over2}\int_\Omega \eps^2(\nabla\psi^\eps_1 \partial_t \overline{\psi^\eps_1}+\nabla\overline{\psi^\eps_1}\partial_t \psi^\eps_1 )\nabla\chi\nnn\\
&\qquad\qquad+\eps^{1+\eta}\int_\Omega [(iA\cdot\nabla\psi^\eps_1 )\partial_t \overline{\psi^\eps_1 }-(iA\cdot\nabla\overline{\psi^\eps_1})\partial_t \psi^\eps_1](1-\chi)
+\eps^{2\eta}\int_\Omega |A|^2 \partial_t \rho^\eps_1 \cdot (1-\chi)\nnn\\
&\qquad={1\over2}\int_\Omega \partial_t |\nabla^{\eps,\eta}_A \psi^\eps_1|^2 (1-\chi)-\eps^2\int_\Omega \Re(\nabla\psi^\eps_1 \partial_t \overline{\psi^\eps_1})\nabla\chi\nnn\\
&\qquad\qquad+{1\over2}\eps^{1+\eta}\int_\Omega iA(\nabla\psi^\eps_1\partial_t \overline{\psi^\eps_1}-\nabla\overline{\psi^\eps_1}\partial_t \psi^\eps_1 )(1-\chi)\nnn\\
&\qquad\qquad-{1\over2}\eps^{1+\eta}\int_\Omega i(\partial_t A)(\nabla\psi^\eps_1 \overline{\psi^\eps_1}-\nabla\overline{\psi^\eps_1}\psi^\eps_1)(1-\chi)\nnn\\
&\qquad\qquad-{1\over2}\eps^{1+\eta}\int_\Omega iA[\partial_t (\nabla\psi^\eps_1)\overline{\psi^\eps_1}-\partial_t(\nabla\overline{\psi^\eps_1})\psi^\eps_1](1-\chi)\nnn\\
&\qquad\qquad-{1\over2}\eps^{2\eta}\int_\Omega (\partial_t |A|^2)|\psi^\eps_1|^2 (1-\chi),
\end{align}}
via the integration by parts and the Neumann boundary conditions \eqref{0219-7}.

Since $A$ is divergence-free, i.e. div$A$=0 for $x\in\Omega$ and $t>0$, we have
\begin{equation}\label{divA1}
(A\cdot\nabla\psi^\eps_1 )\partial_t \overline{\psi^\eps_1 }=\text{div}(A\psi^\eps_1 )\partial_t \overline{\psi^\eps_1 }.
\end{equation}
Consequently, we may use integration by parts and $A\cdot\vec{n}\Big|_{\partial\Omega}=0$ for $t>0$ to get
\begin{equation}\label{divA2}
\int_\Omega (A\cdot\nabla\psi^\eps_1 )\partial_t \overline{\psi^\eps_1 }(1-\chi)
=-\int_\Omega A\cdot\nabla(\partial_t \overline{\psi^\eps_1})\psi^\eps_1 (1-\chi)+\int_\Omega A\psi^\eps_1 \partial_t \overline{\psi^\eps_1}\nabla\chi.
\end{equation}
Similarly, we have
\begin{equation}\label{divA3}
\int_\Omega (A\cdot\nabla\overline{\psi^\eps_1} )\partial_t \psi^\eps_1 (1-\chi)=-\int_\Omega A\cdot\overline{\psi^\eps_1}\nabla(\partial_t \psi^\eps_1)(1-\chi)+\int_\Omega A\overline{\psi^\eps_1} \partial_t \psi^\eps_1\nabla\chi.
\end{equation}

Therefore, it follows from \eqref{Aoperator}-\eqref{divA3} that
\begin{align}\label{Aoperatorfull}
&-{1\over2}\int_\Omega (\Delta^{\eps,\eta}_A \psi^\eps_1 \cdot\partial_t \overline{\psi^\eps_1}+\overline{\Delta^{\eps,\eta}_A \psi^\eps_1}\cdot\partial_t \psi^\eps_1)(1-\chi)\nnn\\
&\qquad={1\over2}\int_\Omega \eps^2(\partial_t |\nabla^{\eps}_A \psi^\eps_1 |^2)(1-\chi)-\eps^2\int_\Omega \Re(\nabla\psi^\eps_1 \partial_t \overline{\psi^\eps_1})\nabla\chi\nnn\\
&\qquad\quad+\eps^\eta \int_\Omega (\partial_t A)J^\eps_1 (1-\chi)+\eps^{1+\eta}\int_\Omega A\Im[\overline{\psi^\eps_1}\partial_t \psi^\eps_1]\nabla\chi.
\end{align}

Substituting \eqref{Aoperatorfull} into \eqref{Aoperator}, we get that
\begin{align}\label{id1}
&{1\over2}\int_\Omega \eps^2(\partial_t |\nabla^{\eps,\eta}_A \psi^\eps_1 |^2)(1-\chi)-\eps^2\int_\Omega \Re(\nabla\psi^\eps_1 \partial_t \overline{\psi^\eps_1})\nabla\chi\nnn\\
&\qquad\quad+\eps^\eta\int_\Omega (\partial_t A)J^\eps_1 (1-\chi)+\eps^{1+\eta}\int_\Omega A\Im[\overline{\psi^\eps_1}\partial_t \psi^\eps_1]\nabla\chi\nnn\\
&\qquad\quad+\int_\Omega \partial_t \left[(V+1)\rho^\eps_1)\right] (1-\chi)-\int_\Omega (\partial_t V) \rho^\eps_1 (1-\chi)\nnn\\
&\qquad\quad+\int_\Omega (\rho^\eps_1 +\rho^\eps_2 -1) (\partial_t \rho^\eps_1) (1-\chi) +(\gamma-1)\int_\Omega \rho^\eps_2 (\partial_t \rho^\eps_1) (1-\chi)\nnn\\
&=0.
\end{align}

Therefore, we obtain that
{\allowdisplaybreaks
\begin{align}\label{psi1energy}
&{1\over2}\int_\Omega \left[\partial_t |\nabla^{\eps,\eta}_A \psi^\eps_1 |^2 \right](1-\chi)+\int_\Omega \partial_t \left[(V+1)\rho^\eps_1 \right](1-\chi)\nnn\\
&\qquad\quad+\int_\Omega (\rho^\eps_1 +\rho^\eps_2 -1 )(\partial_t \rho^\eps_1 )(1-\chi)+(\gamma-1)\int_\Omega \rho^\eps_2 (\partial_t \rho^\eps_1 )(1-\chi)\nnn\\
&\qquad=\eps^2\int_\Omega \Re(\nabla\psi^\eps_1 \partial_t \overline{\psi^\eps_1})\nabla\chi-\eps^\eta\int_\Omega (\partial_t A)J^\eps_1 (1-\chi)\nnn\\
&\qquad-\eps^{1+\eta}\int_\Omega A\Im[\overline{\psi^\eps_1}\partial_t \psi^\eps_1]\nabla\chi+\int_\Omega (\partial_t V)\rho^\eps_1 (1-\chi).
\end{align}
}

Similarly, for $\psi^\eps_2$ we have
\begin{align}\label{psi2energy}
&{1\over2}\int_\Omega \left[\partial_t |\nabla^{\eps,\eta}_A \psi^\eps_2 |^2 \right](1-\chi)+\int_\Omega \partial_t \left[(V+1)\rho^\eps_2 \right](1-\chi)\nnn\\
&\qquad\quad+\int_\Omega (\rho^\eps_1 +\rho^\eps_2 -1 )(\partial_t \rho^\eps_2 )(1-\chi)+(\gamma-1)\int_\Omega \rho^\eps_1 (\partial_t \rho^\eps_2 )(1-\chi)\nnn\\
&\qquad=\eps^2\int_\Omega \Re(\nabla\psi^\eps_2 \partial_t \overline{\psi^\eps_2})\nabla\chi-\eps^\eta\int_\Omega (\partial_t A)J^\eps_2 (1-\chi)\nnn\\
&\qquad-\eps^{1+\eta}\int_\Omega A\Im[\overline{\psi^\eps_2}\partial_t \psi^\eps_2]\nabla\chi+\int_\Omega (\partial_t V)\rho^\eps_2 (1-\chi).
\end{align}

Adding \eqref{psi1energy} and \eqref{psi2energy} together, we have
\begin{align}\label{energyt}
&\frac{d}{dt}\left\{{1\over2}\int_\Omega \sum^2_{k=1}|\nabla^{\eps,\eta}_A \psi^\eps_k|^2 (1-\chi)+\int_\Omega \sum^2_{k=1}(V+1)\rho^\eps_k (1-\chi)\right.\nnn\\
&\qquad\quad\left.+{1\over2}\int_\Omega (\rho^\eps_1 +\rho^\eps_2 -1)^2 (1-\chi)+(\gamma-1)\int_\Omega \rho^\eps_1 \rho^\eps_2 (1-\chi)\right\}\nnn\\
&\qquad=\int_\Omega\sum^2_{k=1}\eps^2 \Re(\nabla\psi^\eps_k \partial_t \overline{\psi^\eps_k})\nabla\chi-\int_\Omega\sum^2_{k=1}\eps^\eta(\partial_t A)J^\eps_k (1-\chi)\nnn\\
&\qquad\quad-\eps^{1+\eta}\int_\Omega\sum^2_{k=1} A\Im[\overline{\psi^\eps_k}\partial_t \psi^\eps_k]\nabla\chi+\int_\Omega\sum^2_{k=1} (\partial_t V)\rho^\eps_k (1-\chi).
\end{align}

{\bf Step 2.} Analogously to the process in Step 1, we multiply $\chi(x+u^{\eps,\eta,\infty}t)\partial_t \phi^\eps_1$ to the complex conjugate of \eqref{phi1gteqn} and $\chi(x+u^{\eps,\eta,\infty}t)\partial_t \overline{\phi^\eps_1}$ to \eqref{phi1gteqn}, respectively.
Then, by summing up those resulting equations and choosing the real part, we have that
{\allowdisplaybreaks
\begin{align}\label{energychi}
&\frac{d}{dt}\left[\sum^2_{k=1}\int_\Omega {1\over2} |\nabla^{\eps,\eta}_{\widetilde{A}-A^\infty}\phi^\eps_k (t,x)|^2 \chi(x+u^{\eps,\eta,\infty}t)+\sum^2_{k=1}\int_\Omega (\widetilde{V}-V^\infty )|\phi^\eps_k (t,x)|^2\chi(x+u^{\eps,\eta,\infty}t)\right]\nnn\\
&\qquad\quad+\frac{d}{dt}\left[\int_\Omega {1\over2}(|\phi^\eps_1|^2 +|\phi^\eps_2|^2 -1)^2 \chi(x+u^{\eps,\eta,\infty}t)
+(\gamma-1)\int_\Omega |\phi^\eps_1|^2 |\phi^\eps_2|^2 \chi(x+u^{\eps,\eta,\infty}t)\right]\nnn\\
&\qquad=:\sum^2_{k=1}\sum^{11}_{i=1}I_{i,k}+{1\over2}\int_\Omega (|\phi^\eps_1 |^2 +|\phi^\eps_2 |^2 -1)^2 \cdot\left[\nabla\chi(x+u^{\eps,\eta,\infty}t)\cdot u^{\eps,\eta,\infty}\right]\nnn\\
&\qquad\quad+(\gamma-1)\int_\Omega |\phi^\eps_1 |^2 |\phi^\eps_2 |^2 \cdot\left[\nabla\chi(x+u^{\eps,\eta,\infty}t)\cdot u^{\eps,\eta,\infty}\right],
\end{align}}
where $I_{i,k}$'s ($i=1,\cdots, 11$, $k=1,2$) are defined respectively in the following.

By applying a change of variable $y=x+u^{\eps,\eta,\infty}t$ and a direct computation, we have that
\begin{align}\label{I5k}
I_{1,k}:&={\eps^2\over2}\int_\Omega |\nabla\phi^\eps_k|^2 \cdot\left[\nabla\chi(x+u^{\eps,\eta,\infty}t)\cdot u^{\eps,\eta,\infty}\right]\nnn\\
&={1\over2}\int_{\Omega}\left|\eps\nabla\psi^\eps_k -iU^\infty \psi^\eps_k \right|^2 \cdot\left(\nabla\chi\cdot u^{\eps,\eta,\infty}\right).
\end{align}

\begin{equation}\label{I6k}
I_{2,k}:=\frac{\eps^{2\eta}}{2}\int_\Omega |\widetilde{A}-A^\infty|^2 |\phi^\eps_k |^2 \cdot\left[\nabla\chi(x+u^{\eps,\eta,\infty}t)\cdot u^{\eps,\eta,\infty}\right].
\end{equation}

\begin{align}\label{I7k}
I_{3,k}:&=-\eps^{1+\eta}\int_\Omega (\widetilde{A}-A^\infty )\Im(\nabla\phi^\eps_k \overline{\phi^\eps_k})\cdot\left[\nabla\chi(x+u^{\eps,\eta,\infty}t)\cdot u^{\eps,\eta,\infty}\right]\nnn\\
&=-\int_\Omega (A-A^\infty)\left[\eps^\eta \Im (\nabla^{\eps,\eta}_A \psi^\eps_k \overline{\psi^\eps_k})+\eps^{2\eta}A\rho^\eps_k -\eps^\eta U^\infty \rho^\eps_k \right](\nabla\chi\cdot u^{\eps,\eta,\infty})\nnn\\
&=-\eps^\eta \int_{\Omega}(A-A^\infty )(J^\eps_k -u^{\eps,\eta,\infty}\rho^\eps_k )\cdot(\nabla\chi\cdot u^{\eps,\eta,\infty})-\eps^\eta \int_{\Omega} |A-A^\infty|^2 \rho^\eps_k \cdot(\nabla\chi\cdot u^{\eps,\eta,\infty}).
\end{align}

\begin{align}\label{I8k}
I_{4,k}:&=-\eps^2\int_\Omega \Re(\nabla\phi^\eps_k\partial_t \overline{\phi^\eps_k} )\cdot\nabla\chi(x+u^{\eps,\eta,\infty}t)\nnn\\
&=-\eps^2\int_\Omega \Re\left\{\left(\nabla\psi^\eps_k -i\eps^{-1}U^\infty \psi^\eps_k\right)\cdot\nabla\chi\right.\nnn\\
&\qquad\quad\left.\cdot\left[\partial_t \overline{\psi^\eps_k}+u^{\eps,\eta,\infty}\nabla\overline{\psi^\eps_k}+i\eps^{-1}\left({1\over2}\left(|U^\infty|^2 -\eps^{2\eta}|A^\infty|^2\right)-V^\infty-1-(\gamma-1)a_{k^*}\right)\overline{\psi^\eps_k}\right] \right\},
\end{align}
where $k^*=1$ if $k=2$, $k^*=2$ if $k=1$.

\begin{align}\label{I9k}
I_{5,k}:&=-\eps^{1+\eta}\int_\Omega (\partial_t \widetilde{A})\Im(\nabla\phi^\eps_k \overline{\phi^\eps_k})\cdot\chi(x+u^{\eps,\eta,\infty}t)\nnn\\
&=-\eps^\eta\int_\Omega\sum^2_{j=1} (\partial_t A_j +u^{\eps,\eta,\infty}\cdot\nabla A_j )(J^\eps_k -u^{\eps,\eta,\infty}\rho^\eps_k)\chi\nnn\\
&\qquad-\eps^{2\eta}\int_\Omega\sum^2_{j=1} (\partial_t A_j +u^{\eps,\eta,\infty}\cdot\nabla A_j ) (A-A^\infty)\rho^\eps_k\chi.
\end{align}

\begin{align}\label{I10k}
I_{6,k}:&=\eps^{1+\eta}\int_\Omega (\widetilde{A}-A^\infty )\Im(\overline{\phi^\eps_k}\partial_t \phi^\eps_k )\nabla\chi(x+u^{\eps,\eta,\infty}t)\nnn\\
&=\eps^{1+\eta}\int_\Omega (A-A^\infty)\Im(\overline{\psi^\eps_k}\partial_t \psi^\eps_k)\nabla\chi+\eps^{\eta}\int_\Omega (A-A^\infty)J^\eps_k (u^{\eps,\eta,\infty}\cdot\nabla\chi)\nnn\\
&\qquad+\eps^{2\eta}\int_\Omega (A-A^\infty)A\rho^\eps_k(u^{\eps,\eta,\infty}\cdot\nabla\chi)-\eps^{\eta}\int_\Omega \left[(A-A^\infty)\cdot\nabla\chi\right] M_k \rho^\eps_k,
\end{align}
where $M_k ={1\over2}\left(|U^\infty|^2 -\eps^{2\eta}|A^\infty|^2\right)-V^\infty-1-(\gamma-1)a_{k^*}$, with $k^*=1$ if $k=2$, $k^*=2$ if $k=1$.

\begin{align}\label{I11k}
I_{7,k}:&=\eps^{2\eta}\int_\Omega (\widetilde{A}-A^\infty )(\partial_t \widetilde{A})|\phi^\eps_k |^2 \cdot\chi(x+u^{\eps,\eta,\infty}t)\nnn\\
&=\eps^{2\eta}\int_\Omega\sum^2_{j=1}(A-A^\infty )(\partial_t A_j +u^{\eps,\eta,\infty}\cdot\nabla A_j )|\psi^\eps_k|^2 \chi.
\end{align}

\begin{align}\label{I12k}
I_{8,k}:&=\int_\Omega(\partial_t \widetilde{V})|\phi^\eps_k |^2 \cdot\chi(x+u^{\eps,\eta,\infty}t)\nnn\\
&=\int_\Omega(\partial_t V+u^{\eps,\eta,\infty}\cdot\nabla V)|\psi^\eps_k |^2 \cdot\chi.
\end{align}

\begin{align}\label{I13k}
I_{9,k}:&=\int_\Omega (\widetilde{V}-V^\infty )|\phi^\eps_k |^2 \cdot\left[\nabla\chi(x+u^{\eps,\eta,\infty}t)\cdot u^{\eps,\eta,\infty}\right]\nnn\\
&=\int_\Omega (V-V^\infty )|\psi^\eps_k |^2 \cdot\left(\nabla\chi\cdot u^{\eps,\eta,\infty}\right).
\end{align}

\begin{align}\label{I16k}
I_{10,k}:&=\eps^\eta\int_\Omega (\widetilde{A}-A^\infty )(\partial_t |\phi^\eps_k |^2 )\cdot\left[\chi(x+u^{\eps,\eta,\infty}t)\cdot u^{\eps,\eta,\infty}\right]\nnn\\
&=2\eps^\eta \int_\Omega \sum^2_{k=1}\Re[u^{\eps,\eta,\infty}\cdot(\widetilde{A}-A^\infty)\overline{\phi^\eps_k}\partial_t \phi^\eps_k \chi(x+u^{\eps,\eta,\infty}t)]\nnn\\
&=2\eps^\eta \int_\Omega \sum^2_{k=1}\Re\left[u^{\eps,\eta,\infty}\cdot(\widetilde{A}-A^\infty)[\overline{\psi^\eps_k}(\partial_t \psi^\eps_k +u^{\eps,\eta,\infty}\cdot\nabla\psi^\eps_k)](t,x+u^{\eps,\eta,\infty}t)\chi(x+u^{\eps,\eta,\infty}t)\right]\nnn\\
&=\eps^\eta \int_\Omega \sum^2_{k=1} u^{\eps,\eta,\infty}\cdot(A-A^\infty)\left[\partial_t |\psi^\eps_k |^2 +u^{\eps,\eta,\infty}\cdot\nabla |\psi^\eps_k |^2\right]\chi\nnn\\
&=\eps^\eta \int_\Omega \sum^2_{k=1} \nabla[u^{\eps,\eta,\infty}\cdot(A-A^\infty)\chi]\cdot (J^\eps_k -u^{\eps,\eta,\infty}\rho^\eps_k ),
\end{align}
and
\begin{align}\label{I17k}
I_{11,k}:&=-\int_{\Omega}V^\infty (\partial_t |\phi^\eps_k |^2 )\chi(x+u^{\eps,\eta,\infty}t)\nnn\\
&=-\int_\Omega\sum^2_{k=1}V^\infty (\partial_t |\psi^\eps_k |^2 +u^{\infty}\cdot\nabla|\psi^\eps_k |^2 )\chi\nnn\\
&=-\int_\Omega\sum^2_{k=1}V^\infty (J^\eps_k -u^{\infty}\rho^\eps_k )\cdot\nabla\chi\ ,
\end{align}
where $J^\eps_k$ and $\rho^\eps_k$ are defined as before.

Moreover,
\begin{align}\label{I1k}
\int_\Omega {1\over2} |\nabla^{\eps,\eta}_{\widetilde{A}-A^\infty}\phi^\eps_k (t,x)|^2 \chi(x+u^{\eps,\eta,\infty}t)
&=\int_\Omega {1\over2} |\nabla^{\eps,\eta}_{\widetilde{A}-A^\infty}\phi^\eps_k (t,x)|^2 \chi(x+u^{\eps,\eta,\infty}t)\nnn\\
&=\int_\Omega|\eps\nabla\phi^\eps_k -i\eps^\eta (\widetilde{A}-A^\infty)\phi^\eps_k|^2 \chi(x+u^{\eps,\eta,\infty}t)\nnn\\
&=\int_\Omega|\eps(\nabla\psi^\eps_k -i\eps^{-1}U^\infty \psi^\eps_k )-i\eps^\eta (A-A^\infty )\psi^\eps_k |^2 \chi\nnn\\
&=\int_\Omega|\nabla^{\eps,\eta}_A \psi^\eps_k -iu^{\eps,\eta,\infty}\psi^\eps_k |^2 \chi\ .
\end{align}

Substituting \eqref{I5k}-\eqref{I1k} into \eqref{energychi}, it follows that
 {\allowdisplaybreaks\small\begin{align}\label{dtchi}
 &\frac{d}{dt}\left\{\int_\Omega \sum^2_{k=1}|(\nabla^{\eps,\eta}_A -iu^{\eps,\eta,\infty})\psi^\eps_k |^2 \chi +\int_\Omega \left[{1\over2}(\rho^\eps_1 +\rho^\eps_2 -1)^2 +(\gamma-1)\rho^\eps_1\rho^\eps_2\right]\chi
 +\int_\Omega \sum^2_{k=1} (V-V^\infty )|\psi^\eps_k |^2 \chi\right\}\nnn\\
 &={1\over2}\int_\Omega \sum^2_{k=1} \left|(\eps\nabla-iU^\infty )\psi^\eps_k \right|^2 \cdot(u^{\eps,\eta,\infty}\cdot\nabla\chi)-\frac{\eps^{2\eta}}{2}\int_\Omega \sum^2_{k=1} |A-A^\infty |^2 \rho^\eps_k \cdot (u^{\eps,\eta,\infty}\cdot\nabla\chi)\nnn\\
 &\quad-\eps^\eta \int_\Omega \sum^2_{k=1}(A-A^\infty )(J^\eps_k -u^{\eps,\eta,\infty}\rho^\eps_k ) \cdot (u^{\eps,\eta,\infty}\cdot\nabla\chi)-\eps^2\int_\Omega \sum^2_{k=1}\Re[(\nabla-i\eps^{-1}U^{\infty})\psi^\eps_k\cdot\nabla\chi\partial_t \overline{\psi^\eps_k }]\nnn\\
 &\quad-\eps^2\int_\Omega \sum^2_{k=1}\Re\left\{(\nabla\psi^\eps_k -i\eps^{-1}U^\infty \psi^\eps_k )\cdot\nabla\chi\left[u^{\eps,\eta,\infty}\cdot\nabla\overline{\psi^\eps_k}+i\eps^{-1}\left({1\over2}(|U^\infty|^2 -\eps^{2\eta}|A^\infty|^2 )-V^\infty-1-(\gamma-1){a}_{k^*}\right)\overline{\psi^\eps_k}\right]\right\}\nnn\\
 &\quad\underbrace{-\eps^\eta \int_\Omega \sum^2_{k,j=1}(\partial_t A_j +u^{\eps,\eta,\infty}\nabla A_j )\cdot(J^\eps_k -u^{\eps,\eta,\infty}\rho^\eps_k )\chi}_{=:I_{12}}-\eps^{2\eta} \int_\Omega \sum^2_{k,j=1}(\partial_t A_j +u^{\eps,\eta,\infty}\nabla A_j )\cdot(A-A^\infty )\rho^\eps_k \chi\nnn\\
 &\quad+\eps^{1+\eta}\int_\Omega \sum^2_{k=1}(A-A^\infty )\Im(\overline{\psi^\eps_k}\partial_t \psi^\eps_k )\cdot\nabla\chi\nnn\\
 &\quad+\eps^\eta \int_\Omega \sum^2_{k=1}(A-A^\infty )\left[u^{\eps,\eta,\infty}J^\eps_k +\eps^\eta u^{\eps,\eta,\infty}A\rho^\eps_k - \left({1\over2}(|U^\infty|^2 -\eps^{2\eta}|A^\infty|^2 )-V^\infty-1-(\gamma-1){a}_{k^*}\right)\rho^\eps_k\right]\cdot\nabla\chi\nnn\\
 &\quad+\eps^{2\eta} \int_\Omega \sum^2_{k,j=1}(\partial_t A_j +u^{\eps,\eta,\infty}\nabla A_j )\cdot(A-A^\infty )\rho^\eps_k \chi+\int_\Omega \sum^2_{k=1}\left(\partial_t V +u^{\eps,\eta,\infty}\nabla V\right)\rho^\eps_k \chi\nnn\\
 &\quad+\int_\Omega \sum^2_{k=1} (V-V^\infty )\rho^\eps_k (u^{\eps,\eta,\infty}\cdot\nabla\chi)+{1\over2}\int_\Omega (\rho^\eps_1 +\rho^\eps_2 -1)^2 (u^{\eps,\eta,\infty}\cdot\nabla\chi)\nnn\\
 &\quad+(\gamma-1)\int_\Omega \rho^\eps_1\rho^\eps_2(u^{\eps,\eta,\infty}\cdot\nabla\chi)+\underbrace{\eps^\eta\int_{\Omega} \sum^2_{k,j=1}\partial_j [u^{\eps,\eta,\infty}(A-A^\infty)\chi](J^\eps_k -u^{\eps,\eta,\infty}\rho^\eps_k )}_{=:I_{13}}\nnn\\
 &\quad-\int_\Omega\sum^2_{k=1}V^\infty (J^\eps_k -u^{\infty}\rho^\eps_k )\cdot\nabla\chi\ .
 \end{align}}
{\allowdisplaybreaks
Therefore,
\begin{align}\label{curlA}
I_{12}+I_{13}&=-\eps^\eta \int_\Omega \sum^2_{k=1}(\partial_t A)\cdot(J^\eps_k -u^{\eps,\eta,\infty}\rho^\eps_k )\chi\nnn\\
&\quad\qquad-\eps^\eta \int_\Omega \sum^2_{k=1}\text{curl}A(u^{\eps,\eta,\infty})^\bot(J^\eps_k -u^{\eps,\eta,\infty}\rho^\eps_k )\chi\nnn\\
&\quad\qquad+\eps^\eta \int_\Omega \sum^2_{k=1}u^{\eps,\eta,\infty} (A-A^\infty )(J^\eps_k -u^{\eps,\eta,\infty}\rho^\eps_k )\nabla\chi,
\end{align}}
since $\partial_j A-\nabla A_j =(-1)^{j^*}\text{curl}A\cdot e_{j^*}$, where $\{e_1, e_2\}$ is the standard basis of $\mathbb{R}^2$, and $j^* =2$ if $j=1$, $j^* =1$ if $j=2$.

Substituting \eqref{curlA} into \eqref{dtchi} and with direct computation, we get that
{\small\allowdisplaybreaks
\begin{align}\label{enegychito}
&\frac{d}{dt}\left\{\int_\Omega \sum^2_{k=1}|(\nabla^{\eps,\eta}_A -iu^{\eps,\eta,\infty})\psi^\eps_k |^2 \chi +\int_\Omega \left[{1\over2}(\rho^\eps_1 +\rho^\eps_2 -1)^2 +(\gamma-1)\rho^\eps_1\rho^\eps_2\right]\chi
+\int_\Omega \sum^2_{k=1} (V-V^\infty )|\psi^\eps_k |^2 \chi\right\}\nnn\\
&=\underbrace{{1\over2}\int_\Omega \sum^2_{k=1} \left|(\eps\nabla-iU^\infty )\psi^\eps_k \right|^2 \cdot(u^{\eps,\eta,\infty}\cdot\nabla\chi)}_{=:I_{14}}\underbrace{-\eps^2\int_\Omega \sum^2_{k=1}\Re[(\nabla-i\eps^{-1}U^{\infty})\psi^\eps_k\cdot\nabla\chi\partial_t \overline{\psi^\eps_k }]}_{=:I_{15}}\nnn\\
&\quad\underbrace{-\eps^2\int_\Omega \sum^2_{k=1}\Re\left\{(\nabla\psi^\eps_k -i\eps^{-1}U^\infty \psi^\eps_k )\cdot\nabla\chi\left[u^{\eps,\eta,\infty}\cdot\nabla\overline{\psi^\eps_k}+i\eps^{-1}\left({1\over2}(|U^\infty|^2 -\eps^{2\eta}|A^\infty|^2 )-V^\infty-1-(\gamma-1){a}_{k^*}\right)\overline{\psi^\eps_k}\right]\right\}}_{=:I_{16}}\nnn\\
&\quad-\eps^\eta \int_\Omega \sum^2_{k=1}\left\{\left[\partial_t A+\text{curl}A\cdot(u^{\eps,\eta,\infty})^\bot\right]\chi-u^{\eps,\eta,\infty}\cdot(A-A^\infty )\nabla\chi\right\}(J^\eps_k -u^{\eps,\eta,\infty}\rho^\eps_k )\nnn\\
&\quad+\eps^\eta \int_\Omega \sum^2_{k=1}(A-A^\infty )\left\{u^{\eps,\eta,\infty}J^\eps_k +\left[\eps^\eta u^{\eps,\eta,\infty}(A-A^\infty )-{1\over2}|u^{\eps,\eta,\infty}|^2 +V^\infty+1+(\gamma-1){a}_{k^*}\right]\rho^\eps_k\right\}\cdot\nabla\chi\nnn\\
&\quad-\eps^\eta \int_\Omega \sum^2_{k=1}(A-A^\infty )\left[J^\eps_k -u^{\eps,\eta,\infty}\rho^\eps_k +\frac{\eps^\eta}{2}(A-A^\infty )\rho^\eps_k\right](u^{\eps,\eta,\infty}\cdot\nabla\chi)\nnn\\
&\quad\underbrace{-\eps^{1+\eta}\int_\Omega \sum^2_{k=1}(A-A^\infty )\Im(\psi^\eps_k\partial_t \overline{\psi^\eps_k} )\cdot\nabla\chi}_{=:I_{17}}+\int_\Omega \sum^2_{k=1}\left(\partial_t V +u^{\eps,\eta,\infty}\nabla V\right)\rho^\eps_k \chi\nnn\\
&\quad+\int_\Omega \sum^2_{k=1} (V-V^\infty )\rho^\eps_k (u^{\eps,\eta,\infty}\cdot\nabla\chi)+{1\over2}\int_\Omega (\rho^\eps_1 +\rho^\eps_2 -1)^2 (u^{\eps,\eta,\infty}\cdot\nabla\chi)\nnn\\
&\quad+(\gamma-1)\int_\Omega \rho^\eps_1\rho^\eps_2(u^{\eps,\eta,\infty}\cdot\nabla\chi)-\int_\Omega\sum^2_{k=1}V^\infty (J^\eps_k -u^{\infty}\rho^\eps_k )\cdot\nabla\chi\ .
\end{align}}
Then we have
\begin{equation}\label{1term}
I_{14}={\eps^2\over2}\int_\Omega \sum^2_{k=1}|\nabla\psi^\eps_k |^2 u^{\eps,\eta,\infty}\cdot\nabla\chi-\int_\Omega \sum^2_{k=1} U^\infty \cdot\left[J^\eps_k +\left(-{1\over2}U^\infty +\eps^\eta A\right)\rho^\eps_k \right](u^{\eps,\eta,\infty}\cdot\nabla\chi),
\end{equation}
and
{\small\begin{align}\label{2term}
	I_{16}&=-\eps^2\int_\Omega \sum^2_{k=1}\Re[(\nabla\psi^\eps_k \cdot\nabla\chi)u^{\eps,\eta,\infty}\cdot\nabla\overline{\psi^\eps_k}]+\int_\Omega \sum^2_{k=1}(U^\infty \cdot\nabla\chi)u^{\eps,\eta,\infty}\cdot(J^\eps_k +\eps^\eta A\rho^\eps_k )\nnn\\
	&\quad\quad+\int_\Omega \sum^2_{k=1}\left[{1\over2}(|U^\infty|^2 -\eps^{2\eta}|A^\infty|^2 )-V^\infty-1-(\gamma-1){a}_{k^*}\right]\left[J^\eps_k -u^{\eps,\eta,\infty}\rho^\eps_k +\eps^\eta(A-A^\infty)\rho^\eps_k\right]\cdot\nabla\chi\ .
	\end{align}}
By \eqref{transftrap}, we have that
{\allowdisplaybreaks
\begin{align}\label{3term}
I_{15} +I_{17} &=-\eps^2\int_\Omega \sum^2_{k=1}\Re[\nabla\psi^\eps_k \cdot\nabla\chi\partial_t \overline{\psi^\eps_k }]-\eps^{1+\eta}\int_\Omega \sum^2_{k=1}\Im(A\cdot\nabla\chi\psi^\eps_k \partial_t \overline{\psi^\eps_k })\nnn\\
&\qquad-{{\eps^2}\over2}\int_\Omega \sum^2_{k=1}|\nabla\psi^\eps_k |^2 u^{\eps,\eta,\infty}\cdot\nabla\chi+{{\eps^2}\over4}\int_\Omega \sum^2_{k=1}\Delta(u^{\eps,\eta,\infty}\cdot\nabla\chi)\rho^\eps_k\nnn\\
&\qquad+\eps^\eta\int_\Omega \sum^2_{k=1}(u^{\eps,\eta,\infty}\cdot\nabla\chi)A\cdot\left[J^\eps_k +{{\eps^\eta}\over2}A\rho^\eps_k\right]-\int_\Omega \sum^2_{k=1}V\rho^\eps_k (u^{\eps,\eta,\infty}\cdot\nabla\chi)\nnn\\
&\qquad-\int_\Omega \sum^2_{k=1}(u^{\eps,\eta,\infty}\cdot\nabla\chi)\left(|\rho^\eps_1 |^2 +|\rho^\eps_2 |^2 +2\gamma\rho^\eps_1 \rho^\eps_2\right).
\end{align}}
{\allowdisplaybreaks
{\bf Step 3.} We observe that
\begin{align}\label{4term}
&-\int_\Omega \sum^2_{k=1} U^\infty \cdot\left[J^\eps_k +\left(-{1\over2}U^\infty +\eps^\eta A\right)\rho^\eps_k \right](u^{\eps,\eta,\infty}\cdot\nabla\chi)+\eps^\eta\int_\Omega \sum^2_{k=1}(u^{\eps,\eta,\infty}\cdot\nabla\chi)A\cdot\left[J^\eps_k +{{\eps^\eta}\over2}A\rho^\eps_k\right]\nnn\\
&\qquad-\eps^\eta \int_\Omega \sum^2_{k=1}(A-A^\infty )\left[J^\eps_k -u^{\eps,\eta,\infty}\rho^\eps_k +\frac{\eps^\eta}{2}(A-A^\infty )\rho^\eps_k\right](u^{\eps,\eta,\infty}\cdot\nabla\chi)\nnn\\
&=-\int_\Omega \sum^2_{k=1}u^{\eps,\eta,\infty}\cdot\left[J^\eps_k -{1\over2}u^{\eps,\eta,\infty}\rho^\eps_k \right](u^{\eps,\eta,\infty}\cdot\nabla\chi).
\end{align}}
Consequently, by \eqref{enegychito}-\eqref{4term}, we get that
{\small\allowdisplaybreaks
\begin{align}\label{all1}
\frac{d}{dt}e(t)&=-\eps^2\int_\Omega \sum^2_{k=1}\Re[(\nabla\psi^\eps_k \cdot\nabla\chi)u^{\eps,\eta,\infty}\cdot\nabla\overline{\psi^\eps_k}]+{{\eps^2}\over4}\int_\Omega \sum^2_{k=1}\Delta(u^{\eps,\eta,\infty}\cdot\nabla\chi)\rho^\eps_k\nnn\\
&\quad+\int_\Omega \sum^2_{k=1}\left[\eps^\eta (A\cdot\nabla\chi)u^{\eps,\eta,\infty}\cdot(J^\eps_k +\eps^\eta A\rho^\eps_k )\right]\nnn\\
&\quad+\int_\Omega \sum^2_{k=1}\left[\eps^\eta u^{\eps,\eta,\infty}A+{1\over2}|u^{\eps,\eta,\infty}|^2 -1-(\gamma-1){a}_{k^*}\right]J^\eps_k \cdot\nabla\chi\nnn\\
&\quad-\eps^\eta\int_\Omega \sum^2_{k=1}\partial_t A \cdot J^\eps_k (1-\chi)+\left[\partial_t A +\text{curl}A\cdot(u^{\eps,\eta,\infty})^\bot \right]\cdot(J^\eps_k -u^{\eps,\eta,\infty}\rho^\eps_k )\chi\nnn\\
&\quad-\int_\Omega \sum^2_{k=1}V^\infty J^\eps_k \cdot\nabla\chi+\int_\Omega \sum^2_{k=1}(\partial_t V)\rho^\eps_k +\int_\Omega \sum^2_{k=1}(\nabla V\cdot u^{\eps,\eta,\infty})\rho^\eps_k \chi\nnn\\
&\quad-\int_\Omega \left\{{1\over2}\left[(\rho^\eps_1 +\rho^\eps_2)^2 -1\right] +(\gamma-1)\rho^\eps_1 \rho^\eps_2 \right\}\cdot(u^{\eps,\eta,\infty} \cdot\nabla\chi).
\end{align}}
In addition, by a direct calculation, it yields that
\begin{align}
-\eps^2\int_\Omega \sum^2_{k=1}\Re[(\nabla\psi^\eps_k \cdot\nabla\chi)u^{\eps,\eta,\infty}\cdot\nabla\overline{\psi^\eps_k}]
&=-\int_\Omega \sum^2_{k,j,l=1}\Re\left[(\partial^{\eps,\eta}_{A_j}\psi^\eps_k \overline{\partial^{\eps,\eta}_{A_l}\psi^\eps_k})u^{\eps,\eta,\infty}_l \partial_j \chi\right]\nnn\\
&\qquad\quad-\eps^\eta\int_\Omega \sum^2_{k=1}\left[(J^\eps_k \cdot u^{\eps,\eta,\infty})A+(A\cdot u^{\eps,\eta,\infty})(J^\eps_k +\eps^\eta A\rho^\eps_k )\right]\cdot\nabla\chi.\label{1termall1}
\end{align}
Combining \eqref{energyt}, \eqref{all1} and \eqref{1termall1}, we obtain that
{\small\allowdisplaybreaks
\begin{align*}
\frac{d}{dt}e(t)&=-\int_\Omega \sum^2_{k,j,l=1}\Re\left[(\partial^{\eps,\eta}_{A_j}\psi^\eps_k \overline{\partial^{\eps,\eta}_{A_l}\psi^\eps_k})u^{\eps,\eta,\infty}_l \partial_j \chi\right]+{{\eps^2}\over4}\int_\Omega \sum^2_{k=1}\Delta(u^{\eps,\eta,\infty}\cdot\nabla\chi)\rho^\eps_k\nnn\\
&\quad+\int_\Omega \sum^2_{k=1}\left[{1\over2}|u^{\eps,\eta,\infty}|^2 -1-(\gamma-1){a}_{k^*}-V^\infty\right]J^\eps_k \cdot\nabla\chi-\eps^\eta \partial_t A\cdot J^\eps_k (1-\chi)\\
&\quad-\eps^\eta \int_\Omega \sum^2_{k=1}\left[\partial_t A +\text{curl}A\cdot(u^{\eps,\eta,\infty})^\bot \right]\cdot(J^\eps_k -u^{\eps,\eta,\infty}\rho^\eps_k )\chi\\
&\quad+\int_\Omega \sum^2_{k=1}(\partial_t V)\rho^\eps_k +\int_\Omega \sum^2_{k=1}(\nabla V\cdot u^{\eps,\eta,\infty})\rho^\eps_k \chi\nnn\\
&\quad-\int_\Omega \left\{{1\over2}\left[(\rho^\eps_1 +\rho^\eps_2)^2 -1\right] +(\gamma-1)\rho^\eps_1 \rho^\eps_2 \right\}\cdot(u^{\eps,\eta,\infty} \cdot\nabla\chi),
\end{align*}}
which gives us the desired identity of part (i).

\vskip0.3cm

(ii) The proof of this part is similar to the one in \cite{ll08}, we skip it here.

(iii) The proof is similar to the approach in \cite{ll08} when $\alpha=0$. Multiplying $\partial^{\eps,\eta}_{A_j}\psi^\eps_1$ to the complex conjugate of $\psi^\eps_1$-equation in \eqref{transftrap}, we then have that
\begin{multline*}
-i\eps\partial_t \overline{\psi^\eps_1}\cdot\partial^{\eps,\eta}_{A_j}\psi^\eps_1=-{1\over2}\overline{\Delta^{\eps,\eta}_A \psi^\eps_1}\cdot\partial^{\eps,\eta}_{A_j}\psi^\eps_1+V\overline{\psi^\eps_1}\cdot\partial^{\eps,\eta}_{A_j}\psi^\eps_1 +\overline{\psi^\eps_1}\cdot\partial^{\eps,\eta}_{A_j}\psi^\eps_1 \\
+(|\psi^\eps_1|^2 +|\psi^\eps_2 |^2 -1)\overline{\psi^\eps_1}\cdot\partial^{\eps,\eta}_{A_j}\psi^\eps_1 +(\gamma-1)|\psi^\eps_2 |^2 \overline{\psi^\eps_1}\cdot\partial^{\eps,\eta}_{A_j}\psi^\eps_1.
\end{multline*}

On the other hand, taking $\partial^{\eps,\eta}_{A_j}$ on the $\psi^\eps_1$-equation in \eqref{transftrap} and multiplying the resulting equation by $\overline{\psi^\eps_1}$, we obtain that
\begin{multline*}
i\eps\overline{\psi^\eps_1}\partial^{\eps,\eta}_{A_j}(\partial_t \psi^\eps_1 )=-{1\over2}\overline{\psi^\eps_1}\partial^{\eps,\eta}_{A_j}(\Delta^{\eps,\eta}_A \psi^\eps_1 )+\overline{\psi^\eps_1}\partial^{\eps,\eta}_{A_j}(V\psi^\eps_1 )+\overline{\psi^\eps_1}(\partial^{\eps,\eta}_{A_j}\psi^\eps_1 )\\
+\overline{\psi^\eps_1}\partial^{\eps,\eta}_{A_j}[(|\psi^\eps_1|^2 +|\psi^\eps_2 |^2 -1)\psi^\eps_1 ]+(\gamma-1)\overline{\psi^\eps_1}\partial^{\eps,\eta}_{A_j}(|\psi^\eps_2|^2 \psi^\eps_1).
\end{multline*}

Adding the above two identities together and choosing the real part of the resulting equation, it follows that
\begin{align}\label{J1}
\partial_t J^\eps_{1,j}+{1\over2}\eps^{-1}\Re[\overline{\Delta^{\eps,\eta}_A \psi^\eps_1}\cdot\partial^{\eps,\eta}_{A_j}\psi^\eps_1& -\overline{\psi^\eps_1}\partial^{\eps,\eta}_{A_j}(\Delta^{\eps,\eta}_A \psi^\eps_1 )]\notag\\
&+(\partial_j V)\rho^\eps_1 +{1\over2}\partial_j |\rho^\eps_1 |^2 +\gamma\rho^\eps_1 (\partial_j \rho^\eps_2 )+ (\partial_t A_j )\rho^\eps_1=0.
\end{align}
Applying Claim 1 in the proof of Lemma 2.1 in \cite{ll08} with $\alpha=0$, it yields that
\begin{align}\label{claim}
\eps^{-1}\Re[\overline{\Delta^{\eps,\eta}_A \psi^\eps_1}\cdot\partial^{\eps,\eta}_{A_j}\psi^\eps_1& -\overline{\psi^\eps_1}\partial^{\eps,\eta}_{A_j}(\Delta^{\eps,\eta}_A \psi^\eps_1 )]\notag\\
&=2\sum^2_{l=1}\partial_l \Re[\overline{\partial^{\eps,\eta}_{A_l}\psi^\eps_1}\cdot\partial^{\eps,\eta}_{A_j}\psi^\eps_1]
-{{\eps^2}\over2}\partial_j (\Delta\rho^\eps_1 )+2(-1)^{j} \text{curl}A\cdot J^\eps_{1,j^*},
\end{align}
where $j^* =2$ if $j=1$, and $j^* =1$ if $j=2$.
Substituting \eqref{claim} into \eqref{J1}, we get that
\begin{multline*}
\partial_t J^\eps_{1,j}+\sum^2_{l=1}\partial_l \left[\Re(\overline{\partial^{\eps,\eta}_{A_l}\psi^\eps_1}\cdot\partial^{\eps,\eta}_{A_j}\psi^\eps_1)\right]-{{\eps^2}\over4}\partial_j (\Delta\rho^\eps_1 )\\
+(-1)^{j} \text{curl}A\cdot J^\eps_{1,j^*} +(\partial_j V)\rho^\eps_1 +{1\over2}\partial_j |\rho^\eps_1 |^2 +\gamma\rho^\eps_1 \partial_j \rho^\eps_2 + (\partial_t A_j )\rho^\eps_1 =0.
\end{multline*}
Similarly, we have the analogous identity for $\psi^\eps_2$
\begin{multline*}
\partial_t J^\eps_{2,j}+\sum^2_{l=1}\partial_l \left[\Re(\overline{\partial^{\eps,\eta}_{A_l}\psi^\eps_2}\cdot\partial^{\eps,\eta}_{A_j}\psi^\eps_2)\right]-{{\eps^2}\over4}\partial_j (\Delta\rho^\eps_2 )\\
+(-1)^{j} \text{curl}A\cdot J^\eps_{2,j^*} +(\partial_j V)\rho^\eps_2 +{1\over2}\partial_j |\rho^\eps_2 |^2 +\gamma\rho^\eps_2 \partial_j \rho^\eps_1 + (\partial_t A_j )\rho^\eps_2 =0.
\end{multline*}
Adding the above two identities together, we have that
\begin{multline}
\sum^2_{k=1}\partial_t J^\eps_{k, j}+\sum^2_{k,l=1}\partial_l \left[\Re(\overline{\partial^{\eps,\eta}_{A_l}\psi^\eps_k}\cdot\partial^{\eps,\eta}_{A_j}\psi^\eps_k)\right]-{{\eps^2}\over4}\sum^2_{k=1}\partial_j (\Delta\rho^\eps_k )
+\sum^2_{k=1}(\partial_j V)\rho^\eps_k\\ +{1\over2}\sum^2_{k=1}\partial_j |\rho^\eps_k |^2 +\gamma\partial_j (\rho^\eps_1 \rho^\eps_2 )+\sum^2_{k=1} (\partial_t A_j )\rho^\eps_k
+\sum^2_{k=1}(-1)^{j} \text{curl}A\cdot J^\eps_{k,j^*} =0.
\end{multline}
The proof is complete.

\end{proof}

\section{Proof of Theorem \ref{main}}

In this section, we introduce the modulated energy functional \eqref{Hfunct} as we mentioned before, which will be called the $H$-function. $(\psi^\eps_1, \psi^\eps_2 )$ is the solution for the system \eqref{trap} with the Neumann boundary condition \eqref{0219-7}, and $(\rho, u)$ is the solution of \eqref{curlfield}-\eqref{sbc}, respectively. We can also rewrite $H$-function \eqref{Hfunct} as below:
\begin{equation}
{H}(t)={{\eps^2}\over2}\int_{\Omega}\sum^2_{k=1}\left|\nabla\sqrt{\rho^\eps_k}\right|^2 dx+{1\over2}\int_\Omega \sum^{2}_{k=1}\frac{1}{\rho^\eps_k}\left|J^\eps_k -\rho^\eps_k u\right|^2 dx+\int_\Omega \left[{1\over2}(\rho^\eps_1 +\rho^\eps_2 -\rho)^2 +(\gamma-1)\rho^\eps_1 \rho^\eps_2 \right]dx,\label{haha11241106}
 \end{equation}
which gives us a clear view of that how the $H$-function controls the propagation of densities and linear momenta under the effect of rotating fields and trap potentials.

With the cut-off function $\chi$ defined in \eqref{chifctn} and setting $R\ge R_0$, 
it is obvious that
\begin{align}\label{Hexpansion}
H(t)&=e(t)-\int_\Omega \sum^2_{k=1}[(V+1)(1-\chi)+(V-V^\infty)\chi]\rho^\eps_k
-\int_\Omega \sum^2_{k=1}(u-u^{\eps,\eta,\infty}\chi)(J^\eps_k -u^{\eps,\eta,\infty}\rho^\eps_k )\nnn\\
&\qquad+{1\over2}\int_\Omega \sum^2_{k=1}(|u|^2 -2u\cdot u^{\eps,\eta,\infty}+|u^{\eps,\eta,\infty}|^2 \chi)\rho^\eps_k
+\int_\Omega \left\{{1\over2}(\rho-1)^2 -\sum^2_{k=1}(\rho^\eps_k -a_k )(\rho-1)\right\}.
\end{align}

Now we may use Lemma~\ref{laws} to calculate the time derivative of \eqref{Hexpansion} term by term.
\begin{align}\label{termV}
&-\frac{d}{dt}\int_\Omega \sum^2_{k=1}[(V+1)(1-\chi)+(V-V^\infty)\chi]\rho^\eps_k\nnn\\
&\qquad=-\int_\Omega\sum^2_{k=1}\Big\{\rho^\eps_k (\partial_t V)+\left[(V+1)(1-\chi)+(V-V^\infty)\chi\right]\partial_t \rho^\eps_k\Big\}\nnn\\
&\qquad=-\int_\Omega\sum^2_{k=1}\Big\{\rho^\eps_k (\partial_t V)+J^\eps_k \cdot\nabla V-(V^\infty +1)J^\eps_k \cdot\nabla\chi\Big\}.
\end{align}
Here we used the integration by parts and the Neumann boundary conditions \eqref{0219-7}.

By \eqref{mass}-\eqref{momentum} and the Neumann boundary conditions \eqref{0219-7}, then integrating by parts, we have the following
{\allowdisplaybreaks\begin{align}\label{Jk}
&-\frac{d}{dt}\int_\Omega \sum^2_{k=1}(u-u^{\eps,\eta,\infty}\chi)(J^\eps_k -u^{\eps,\eta,\infty}\rho^\eps_k )\nnn\\
&\qquad=-\int_\Omega \sum^2_{k=1}(\partial_t u)(J^\eps_k -u^{\eps,\eta,\infty}\rho^\eps_k )
-\int_\Omega \sum^2_{k=1}(u-u^{\eps,\eta,\infty}\chi)(\partial_t J^\eps_k -u^{\eps,\eta,\infty}\partial_t \rho^\eps_k )\nnn\\
&\qquad=-\int_\Omega \sum^2_{k=1}(\partial_t u)(J^\eps_k -u^{\eps,\eta,\infty}\rho^\eps_k )
+\int_\Omega \sum^2_{j,k=1}(u_j -u^{\eps,\eta,\infty}_j\chi)\left\{\sum^2_{l=1}\partial_l \Re(\overline{\partial^{\eps,\eta}_{A_l}\psi^\eps_k}\partial^{\eps,\eta}_{A_j}\psi^\eps_k)\right.\nnn\\
&\qquad\quad-{{\eps^2}\over4}\partial_j (\Delta\rho^\eps_k )+\partial_j \left[{1\over2}(\rho^\eps_1 +\rho^\eps_2 )^2 +(\gamma-1)\rho^\eps_1 \rho^\eps_2\right]+\rho^\eps_k \partial_j V\nnn\\
&\qquad\quad+ \rho^\eps_k \partial_t A_j +(-1)^j  {\text{curl}}A\cdot J^\eps_{k,j^*}-u^{\eps,\eta,\infty}_j \text{div}J^\eps_k \Big\}\nnn\\
&\qquad=-\int_\Omega \sum^2_{k=1}(\partial_t u)(J^\eps_k -u^{\eps,\eta,\infty}\rho^\eps_k )
-\int_\Omega \sum^2_{j,k,l=1}\partial_l (u_j -u^{\eps,\eta,\infty}_j\chi)\left[\Re(\overline{\partial^{\eps,\eta}_{A_l}\psi^\eps_k}\partial^{\eps,\eta}_{A_j}\psi^\eps_k)-u^{\eps,\eta,\infty}_j J^\eps_{k,l}\right]\nnn\\
&\qquad\quad-{{\eps^2}\over4}\int_\Omega \sum^2_{k=1}\nabla\text{div}(u-u^{\eps,\eta,\infty}\chi)\cdot\nabla\rho^\eps_k\nnn\\
&\qquad\quad -\int_\Omega \text{div}(u-u^{\eps,\eta,\infty}\chi)\left[\frac{(\rho^\eps_1 +\rho^\eps_2 )^2 -1}{2}+(\gamma-1)\rho^\eps_1 \rho^\eps_2  \right]\nnn\\
&\qquad\quad+\int_\Omega \sum^2_{k=1}(u-u^{\eps,\eta,\infty}\chi)(\nabla V+ \partial_t A)\rho^\eps_k
-\int_\Omega \sum^2_{k=1} \text{curl}A(u-u^{\eps,\eta,\infty}\chi)^{\bot}\cdot J^\eps_k.
\end{align}}
Moreover, as the calculation in \cite{lz06}, we have
\begin{align}\label{term3}
&{1\over2}\frac{d}{dt}\int_\Omega \sum^2_{k=1}(|u|^2 -2u\cdot u^{\eps,\eta,\infty}+|u^{\eps,\eta,\infty}|^2 \chi)\rho^\eps_k\nnn\\
&\qquad={1\over2}\int_\Omega \sum^2_{k=1}J^\eps_k \cdot\nabla(|u|^2 -2u\cdot u^{\eps,\eta,\infty}+|u^{\eps,\eta,\infty}|^2 \chi)
+\int_\Omega \sum^2_{k=1}\rho^\eps_k (u-u^{\eps,\eta,\infty})\partial_t u,
\end{align}
and
\begin{align}\label{term4}
&\frac{d}{dt}\int_\Omega \left\{{1\over2}(\rho-1)^2 -\sum^2_{k=1}(\rho^\eps_k -a_k )(\rho-1)\right\}\nnn\\
&\qquad=-{1\over2}\int_\Omega (\rho^2 -1)\text{div}u+\int_\Omega \sum^2_{k=1}\nabla\rho\cdot(\rho^\eps_k u-J^\eps_k )+\int_\Omega\big[\rho(\rho^\eps_1 +\rho^\eps_2 )-1\big]\text{div}u.
\end{align}
{\allowdisplaybreaks
Combining \eqref{energy} and \eqref{Hexpansion}-\eqref{term4}, we have that
\begin{align}
\frac{d}{dt}H(t)&=-\int_\Omega\sum^2_{j,k,l=1}\partial_l u_j\Re\left(\partial^{\eps,\eta}_{A_j}\psi^\eps_k \overline{\partial^{\eps,\eta}_{A_l}\psi^\eps_k}\right)+2\int_\Omega\sum^2_{k=1}(u\cdot\nabla u)J^\eps_k -\int_\Omega\sum^2_{k=1}(u\cdot\nabla u)\cdot u\rho^\eps_k\nnn\\
&\quad -2\int_\Omega\sum^2_{k=1}(u\cdot\nabla u)J^\eps_k +\int_\Omega\sum^2_{k=1}(u\cdot\nabla u)\cdot u\rho^\eps_k
+\int_\Omega\sum^2_{k=1}\left[(\nabla u\cdot u^{\eps,\eta,\infty})J^\eps_k -|u^{\eps,\eta,\infty}|^2 \cdot\nabla\chi\cdot J^\eps_k\right]\nnn\\
&\quad-{{\eps^2}\over4}\int_\Omega \sum^2_{k=1}(\nabla\text{div}u)\nabla\rho^\eps_k -\int_\Omega{1\over2}(\rho^\eps_1 +\rho^\eps_2 -\rho)^2 (\text{div}u)\nnn\\
&\quad-(\gamma-1)\int_\Omega \sum^2_{k=1}\rho^\eps_1 \rho^\eps_2 (\text{div}u)
+\int_\Omega\sum^2_{k=1}|u^{\eps,\eta,\infty}|^2 (J^\eps_k \cdot\nabla\chi)\nnn\\
&\quad+ \int_\Omega\sum^2_{k=1}\left[\partial_t A+\text{curl}A\cdot u^\bot \right](u\rho^\eps_k -J^\eps_k )
+\int_\Omega\sum^2_{k=1}(\partial_t u +\nabla V+\nabla\rho)(u\rho^\eps_k -J^\eps_k )\nnn\\
&\quad+\int_\Omega\sum^2_{k=1}(u\cdot\nabla u)J^\eps_k -\int_\Omega\sum^2_{k=1}(\nabla u\cdot u^{\eps,\eta,\infty})J^\eps_k\nnn\\
&=-\int_\Omega\sum^2_{j,k,l=1}(\partial_l u_j)\Re\left[(\partial^{\eps,\eta}_{A_j}\psi^\eps_k -iu_j \psi^\eps_k ) (\overline{\partial^{\eps,\eta}_{A_l}\psi^\eps_k -iu_l \psi^\eps_k})\right]-{{\eps^2}\over4}\int_\Omega \sum^2_{k=1}(\nabla\text{div}u)\nabla\rho^\eps_k\nnn\\
&\quad-\int_\Omega{1\over2}(\rho^\eps_1 +\rho^\eps_2 -\rho)^2 (\text{div}u)-(\gamma-1)\int_\Omega \sum^2_{k=1}\rho^\eps_1 \rho^\eps_2 (\text{div}u),
\end{align}}
since $u\cdot u^\bot =0$ and \eqref{curlfield}, i.e.,
 \begin{align}\label{dH}
 \frac{d}{dt}H(t)&=-\int_\Omega\sum^2_{j,k,l=1}(\partial_l u_j)\Re\left[(\partial^{\eps,\eta}_{A_j}\psi^\eps_k -iu_j \psi^\eps_k ) (\overline{\partial^{\eps,\eta}_{A_l}\psi^\eps_k -iu_l \psi^\eps_k})\right]-{{\eps^2}\over4}\int_\Omega \sum^2_{k=1}(\nabla\text{div}u)\nabla\rho^\eps_k\nnn\\
&\qquad-\int_\Omega{1\over2}(\rho^\eps_1 +\rho^\eps_2 -\rho)^2 (\text{div}u)-(\gamma-1)\int_\Omega \sum^2_{k=1}\rho^\eps_1 \rho^\eps_2 (\text{div}u).
 \end{align}

Now, we prove \eqref{converge} for $\eta=0$.
From the analysis above, we need to estimate some terms of \eqref{dH}:
{\allowdisplaybreaks\begin{align}\label{est1}
&\left|-{{\eps^2}\over4}\int_\Omega \sum^2_{k=1}(\nabla\text{div}u)\nabla\rho^\eps_k\right|\nnn\\
&\qquad=\left|-\frac{\eps}{2}\int_\Omega \sum^2_{k=1}(\nabla\text{div}u)\Re[\overline{\psi^\eps_k}\cdot(\nabla_A -iu)\psi^\eps_k ]\right|\nnn\\
&\qquad\le\frac{\eps}{4}\int_\Omega \sum^2_{k=1}|(\nabla_A -iu)\psi^\eps_k |^2 +\frac{\eps}{4}\int_\Omega |\nabla\text{div}u|^2 (\rho^\eps_1 +\rho^\eps_2 -\rho+\rho)\nnn\\
&\qquad\le\frac{\eps}{4}\int_\Omega \sum^2_{k=1}|(\nabla_A -iu)\psi^\eps_k |^2 +\frac{\eps}{4}\int_\Omega |\nabla\text{div}u|^4 +\frac{\eps}{2}\int_\Omega(\rho^\eps_1 +\rho^\eps_2 -\rho)^2 +\rho^2\nnn\\
&\qquad\le C\eps\left[H(t)+\|u\|^4_{H^3 (\Omega)}+\|\rho\|^2_{L^2 (\Omega)}\right],
\end{align}}
where $\nabla_A =\eps\nabla-iA$.
Hence, together with \eqref{dH}, it follows that
\begin{equation*}
\frac{d}{dt}H(t)\le CH(t)+O(\eps).
\end{equation*}
Applying the Gronwall's inequality to the above inequality, we obtain
\begin{equation*}
H(t)\le C(H(0)+O(\eps)).
\end{equation*}
By the assumption (A2) in the introduction, we conclude that
{\allowdisplaybreaks\begin{align}
H(0)&={1\over2}\int_\Omega \sum^2_{k=1}|(\nabla_{A_0} -iu_0 )\psi^\eps_{k,0} |^2 +{1\over2}\int_\Omega (\rho^\eps_{1,0} +\rho^\eps_{2,0} -\rho_0 )^2
+(\gamma-1)\int_\Omega \rho^\eps_{1,0}\rho^\eps_{2,0}\nnn\\
&={1\over2}\int_\Omega \sum^2_{k=1}\left|\eps\nabla\sqrt{\rho^\eps_{k,0}}-i(u_0 -\nabla S^\eps_{k,0} +A_0 )\sqrt{\rho^\eps_{k,0}}\right|^2\nnn\\
&\qquad+{1\over2}\int_\Omega (\rho^\eps_{1,0} +\rho^\eps_{2,0} -\rho_0 )^2
+(\gamma-1)\int_\Omega \rho^\eps_{1,0}\rho^\eps_{2,0}\nnn\\
&={1\over2}\int_\Omega \sum^2_{k=1}\eps^2\left|\nabla\sqrt{\rho^\eps_{k,0}}\right|^2 +{1\over2}\int_\Omega \sum^2_{k=1} \rho^\eps_{k,0}|u_0 -\nabla S^\eps_{k,0} +A_0|^2\nnn\\
&\qquad+{1\over2}\int_\Omega (\rho^\eps_{1,0} +\rho^\eps_{2,0} -\rho_0 )^2
+(\gamma-1)\int_\Omega \rho^\eps_{1,0}\rho^\eps_{2,0}\nnn\\
&\longrightarrow0\quad\text{as}\ \ \eps\to0,
\end{align}}
for any $t\in [0, T_* )$. Therefore, we have
\begin{equation}\label{Ht0}
H(t)\longrightarrow0\ \ \ \text{as}\ \ \eps\to0, \ \ \text{for}\ \ t\in[0,T_*).
\end{equation}
In particular, both \eqref{haha11241106} and \eqref{Ht0} yield that
\begin{equation}
\rho^\eps_1 (t,x) +\rho^\eps_2 (t,x)-1 \longrightarrow \rho(t,x)-1\ \ \text{in}\ \ L^\infty \left([0,T_* ), L^2 (\Omega)\right)\ \ \text{as}\ \ \eps\to0.
\end{equation}
Moreover,
\begin{equation*}
J^\eps_1 (t,x)+J^\eps_2 (t,x)-(\rho u)(t,x)=\sum^2_{k=1}\Im[(\nabla^{\eps,\eta}_A \psi^\eps_k -iu\psi^\eps_k)\overline{\psi^\eps_k }](t,x)+[(\rho^\eps_1 +\rho^\eps_2 -\rho)u](t,x).
\end{equation*}
Applying Cauchy-Schwarz inequality to the above identity, we obtain that
\begin{equation*}
J^\eps_1 (t,x)+J^\eps_2 (t,x)-(\rho u)(t,x)\longrightarrow0\ \ \text{in}\ \ L^\infty\left([0,T_*), L^1_{\text{loc}}(\Omega)\right)\ \ \text{as}\ \ \eps\to0.
\end{equation*}

Finally, when $\gamma>1$, \eqref{haha11241106} and \eqref{Ht0} immediately imply \eqref{rho1124}.
Therefore we complete the proof of Theorem~\ref{main}.

\begin{remark}\label{rk-0406}
By \eqref{haha11241106} and \eqref{Ht0}, we have $\displaystyle\int_\Omega\frac{1}{\rho^\eps_k}\left|J^\eps_k -\rho^\eps_k u\right|^2 dx\to0$ as $\eps\downarrow0$, $k=1, 2$. Along with the conservation of mass density (see Lemma~\ref{laws}(ii)), we obtain
\begin{align*}
J^\eps_k (t,x)-(\rho_k^{\eps} u)(t,x)\longrightarrow0\ \ \text{in}\ \ L^\infty\left([0,T_*), L^1_{\text{loc}}(\Omega)\right)\ \ \text{as}\ \ \eps\to0.
\end{align*}
\end{remark}

\section{Concluding remarks and further problems}\label{sec-remark}
Thanks to the previous work of Lee and Lin \cite{ll08} and the Galilean transformation for the wave function introduced by Lin and Zhang~\cite{lz05}, we study the semi-classical limit of the Gross--Pitaevskii system~\eqref{trap}  for the rotating two-component BECs in the exterior domain in $\mathbb{R}^2$  with non-zero conditions at infinity. On a formal level, \eqref{trap} resembles a coupled systems of nonlinear Schr\"{o}dinger equations with a small parameter $\eps$ scaled by the Planck's constant. However, the parameter appears in front of derivatives of this model is mathematically singular. To investigate the behavior of mass densities and linear momenta in the semi-classical scaling, we assume that in the system \eqref{trap}, the binary mixture of rotating BECs share the same trap potential $V$ and the same rotating field $A$. We further assume that both phase functions~$S^\eps_{1,0}$ and $S^\eps_{2,0}$ act as a same constant velocity~$U^\infty$. The underlying idea is based upon a new Galilean transformation and a modulated energy functional. Along with the standard argument in \cite{ll08,lz05}, we obtain the convergence to the compressible Euler equation~\eqref{curlfield} with the well-prepared initial data. We stress that such a limiting equation cannot be directly obtained from the standard Madelung transformation. As shown in Theorem~\ref{main}, when $\gamma\geq1$, the propagation of $\rho_1^{\eps}+\rho_2^{\eps}$ and $J_1^{\eps}+J_2^{\eps}$ are controlled by the mass density and linear momentum of this compressible Euler equation, respectively. Also the effect of rotating field on the superfluid in the region far away from the obstacle is precisely described and brings in some interesting phenomenon.

We would like to point out that the underlying argument based on the corresponding conservation laws (cf. Lemma~\ref{laws}) provides a basic understanding on the semi-classical analysis of the Gross--Pitaevskii systems. However, it has a limitation due to the fact that the proper forms of the energy functional \eqref{eg-17-1124} and the modulated energy functional~\eqref{Hfunct} depend strongly on the symmetry of $\psi^\epsilon_1$ and $\psi^\epsilon_2$.  Regarding the general rotating two-component BECs, one can consider the following Gross--Pitaevskii system instead of \eqref{trap}:
\begin{equation}\label{newtrap}
\left\{
\begin{array}{ll}
i\epsilon\partial_t \psi^{\epsilon}_1 =-{1\over2}\Delta^{\epsilon, \eta}_{A_1} \psi^\epsilon_1 +V_1 \psi^\epsilon_1 +|\psi^\epsilon_1 |^2 \psi^\epsilon_1 +\gamma |\psi^\epsilon_2 |^2 \psi^\epsilon_1,&\quad\,\,\, x\in\Omega,\ t>0,\\
[-1em]\\
\\[-1em]
i\epsilon \partial_t \psi^{\epsilon}_2 =-{1\over2}\Delta^{\epsilon, \eta}_{A_2} \psi^\epsilon_2 +V_2 \psi^\epsilon_2 +|\psi^\epsilon_2 |^2 \psi^\epsilon_2 +\gamma |\psi^\epsilon_1 |^2 \psi^\epsilon_2,&\quad\,\,\, x\in\Omega,\ t>0,\\
 \psi^{\epsilon}_k \big|_{t=0} = \psi^{\epsilon}_{k,0}(x):=\sqrt{\rho^\eps_{k,0}(x)}\exp\left({\frac{i}{\eps}S^\eps_{k,0} (x)}\right),\,\,k=1,2,&\quad\,\,\, x\in\Omega,\\
 \rho^{\epsilon}_{k,0}(x)\to {a_k},\quad\mathrm{and}\quad{S}^\eps_{k,0} (x)\to{U}_k^{\infty}\cdot {x}, &\text{as}\,\, |x|\to\infty,\,\, t>0,
 \end{array}
\right.
\end{equation}
where the rotating field $A_k$ satisfies \eqref{0219-2}--\eqref{0219-3} with $(A,A^{\infty})=(A_k,A_k^{\infty})$, the trap potential $V_k$ satisfies \eqref{0219-5} with $(V,V^{\infty})=(V_k,V_k^{\infty})$, and ${U}_k^{\infty}$ is a constant two-vector, $k=1, 2$.
When $\gamma=0$ and $a_k>0$, \eqref{newtrap} is decoupled. Then, under the boundary condition~\eqref{0219-7}, we may follow the similar argument in the proof of Theorem~\ref{main} to obtain the convergence of the equation of $\psi^{\epsilon}_k$ to the compressible Euler equation~\eqref{curlfield} with $(\rho,u)=(\rho_k,u_k)$ and
\begin{equation*}
\rho_k(t,x)\to a_k,\,\,u_k(t,x)\to U_k^\infty-A_k^\infty\,\,\text{as}\,\,|x|\to\infty.
\end{equation*}
When $\gamma>0$, the present work deals with the semi-classical limit of \eqref{newtrap} with $A_1=A_2:=A$, $V_1=V_2:=V$ and $U_1^{\infty}=U_2^{\infty}:=U^{\infty}$, i.e., the model~\eqref{trap}. In particular, our argument can be applied to \eqref{newtrap} with $(\text{curl}A_1,\partial_tA_1,\nabla{V}_1)=(\text{curl}A_2,\partial_tA_2,\nabla{V}_2)$ and $U_1^\infty-A_1^\infty={U}_2^\infty-A_2^\infty$.  However, when one of the following conditions holds:
	\begin{itemize}
	\item[\textbf{(i)}] $(\text{curl}A_1,\partial_tA_1,\nabla{V}_1)\neq(\text{curl}A_2,\partial_tA_2,\nabla{V}_2)$,
	\item[\textbf{(ii)}] $(\text{curl}A_1,\partial_tA_1,\nabla{V}_1)=(\text{curl}A_2,\partial_tA_2,\nabla{V}_2)$ and $U_1^\infty-A_1^\infty\neq{U}_2^\infty-A_2^\infty$,
	\end{itemize}
the rigorous semi-classical analysis of \eqref{newtrap} seems a hard problem whose main difficulty lies in the coupling terms of $\psi^{\epsilon}_k$'s and the fact that the corresponding energy functional has no symmetrical property of $|\psi^{\epsilon}_1|^2+|\psi^{\epsilon}_2|^2$. We will keep working on this problem in a forthcoming project.
\\ \\
{\footnotesize {\bf Acknowledgment.}
This work was initiated while Q. Gao was visiting the NCTS at National Taiwan University. She would like to thank the hospitality of the NCTS. Q. Gao is supported by the NNSF of China under grant 11501231 and 11871386, and the ``Fundamental Research Funds for the Central Universities" under grants 2018IB014. C.-C. Lee and T.-C. Lin were partially supported by the Ministry of Science and Technology of Taiwan under grants 104-2115-M-134-001-MY2, 106-2115-M-002-003 and 107-2115-M-007-004. Finally, the authors would also like to thank two anonymous referees for many precious suggestions on the physical meaning of the model and constructive comments that have contributed to the final version of the manuscript.}

%
%
%

\end{document}